\documentclass[a4paper,12pt]{article} 
\usepackage{graphicx,amssymb,color,amsthm,amsmath}
\newtheorem{theorem}{Theorem}
\usepackage{epsfig,float}
\usepackage{color}
\usepackage{xcolor} 
\usepackage{soul}
\usepackage{subfigure}
\usepackage[ruled, ]{algorithm2e}
\usepackage{booktabs}
\usepackage{xurl}
\usepackage{float}
\usepackage{lineno} 

\usepackage[colorlinks=true,citecolor=blue, linkcolor=blue, urlcolor=blue]{hyperref}

\allowdisplaybreaks
\allowdisplaybreaks
\begin{document}

\centerline{\LARGE \bf
Internet malware propagation: Dynamics and }

\medskip

\centerline{\LARGE \bf
control through SEIRV epidemic model with  }
 \medskip

 \centerline{\LARGE \bf  relapse and intervention}


\vspace*{1cm}

\centerline{\bf Samiran Ghosh$^{a,}$\footnote{Corresponding author. Email: samiran.4pi@csir.res.in} V Anil Kumar$^{a}$}

\vspace{0.5cm}

\centerline{ $^a$ CSIR Fourth Paradigm Institute, Wind Tunnel Road, Bangalore - 560037, India}


\vspace{2cm}

\noindent
{\bf Abstract.}
Malware attacks in today’s vast digital ecosystem pose a serious threat. Understanding malware propagation dynamics and designing effective control strategies are therefore essential. In this work, we propose a generic SEIRV model formulated using ordinary differential equations to study malware spread. We establish the positivity and boundedness of the system, derive the malware propagation threshold, and analyze the local and global stability of the malware-free equilibrium. The separatrix defining epidemic regions in the control space is identified, and the existence of a forward bifurcation is demonstrated. Using normalized forward sensitivity indices, we determine the parameters most influential to the propagation threshold. We further examine the nonlinear dependence of key epidemic characteristics on the transmission rate, including the maximum number of infected, time to peak infection, and total number of infected. We propose a hybrid gradient-based global optimization framework using simulated annealing approach to identify effective and cost-efficient control strategies. Finally, we calibrate the proposed model using infection data from the “Windows Malware Dataset with PE API Calls” and investigated the effect of intervention onset time on averted cases, revealing an exponential decay relationship between delayed intervention and averted cases.

\vspace{1cm}


\medskip

\vspace*{0.5cm}

\setcounter{equation}{0}
\setcounter{section}{0}
\setcounter{page}{1}

\section{Introduction}

The Internet of Things can be defined as a network of smart objects capable of self-organization, information sharing, and adapting to environmental changes \cite{madakam2015internet}. Equipped with specific features like sensors, software, processing units, and communication capabilities, IoT devices have become an integral part of daily human life. The world of IoT devices has been growing exponentially, with applications in diverse fields such as smart homes, wearable technology, healthcare, remote monitoring appliances, transportation systems, and more. The global number of IoT devices has increased exponentially over the last decade. Figure~\ref{statista1} shows the increasing trend in the worldwide number of IoT connected devices as well as the worldwide number of attacks on IoT devices.

\begin{figure}[ht!]
\begin{center}
\includegraphics[scale=0.45]{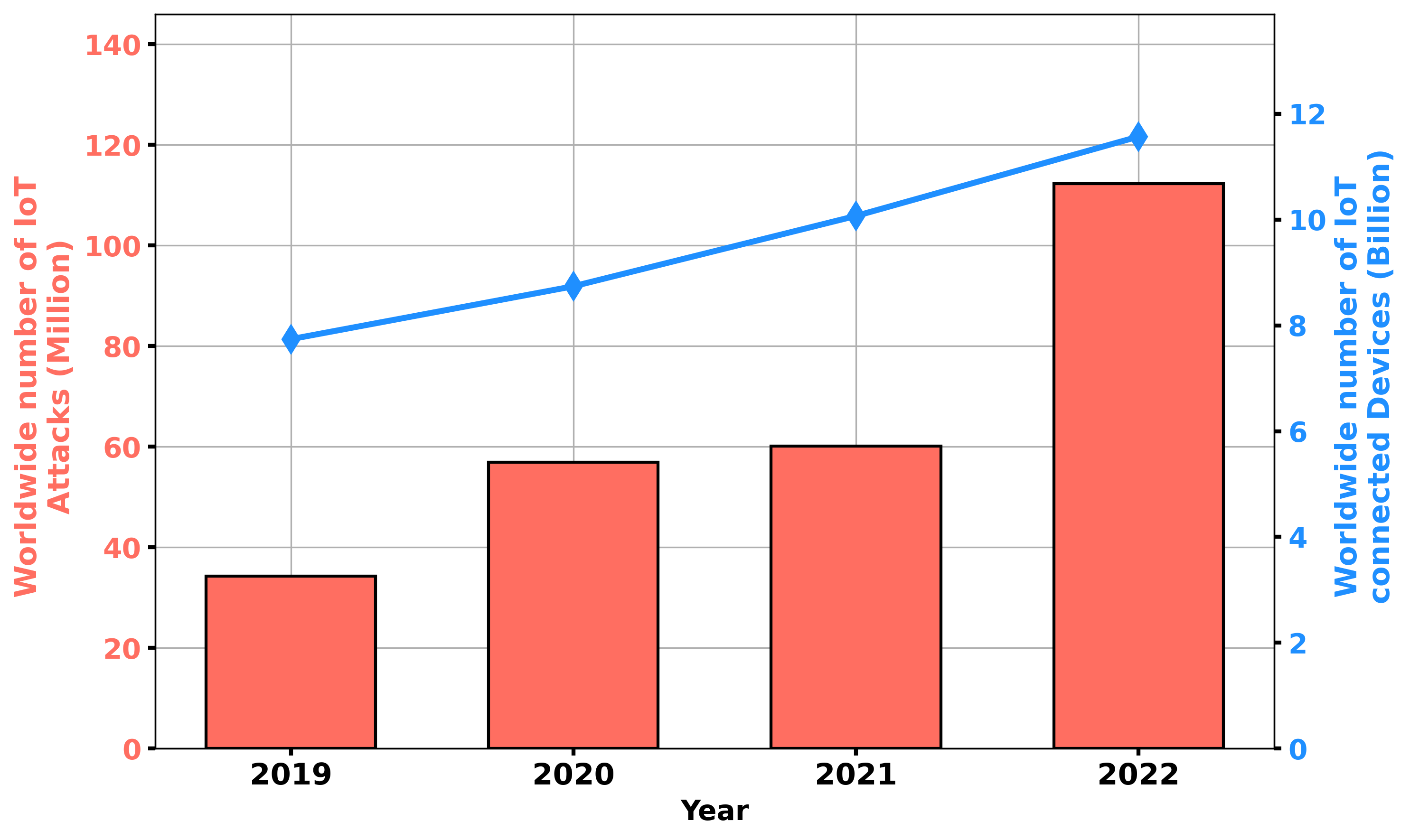}
\caption{The worldwide number of IoT attacks and IoT connected devices are represented by the red bars and blue line, respectively. The red bars correspond to the left vertical axis, while the blue line corresponds to the right vertical axis. The references for these data are \cite{lemevs2023role,danladi2022low}.}
\label{statista1}
\end{center}
\end{figure}

Ensuring the security of this vast and growing digital ecosystem remains a major concern, as cyber attacks continue to pose significant threats. Figure~\ref{statista1} illustrates key statistics on the rising trend of IoT cyber attacks as well as IoT connected population. The vulnerability of these IoT devices arises due to both technological limitations and human behavior. Due to the use of default credentials, lack of strong security features, and lack of routine software updates, the devices can get malware infections. Malware propagation in an IoT network typically begins with an attacker communicating through a Command-and-Control (C\& C) server, which then disseminates malware to infected IoT devices and subsequently spreads to other vulnerable devices (see Figure \ref{flow_iot_1}). Some common types of malware attacks on IoT devices are ransomware, viruses, worms, bots, trojans, etc. \cite{milosevic2016malware}. Through ransomware, the attackers lock the device and demand money to recover the device. Bots try to control the devices remotely and form botnets, which target large-scale attacks like DDoS. Trojans disguise themselves as legitimate programs to convince the user to install them, whereas they are actually malicious. Worms can replicate and spread automatically over networks and infect devices. Viruses infect files or software by injecting malicious, replicable code and spread to other devices through infected files or software. As per the statistics, ransomware attacks occur every 10 seconds, and as of 2023, over 70 percent of companies face this attack \cite{kovtun2024cyber}. These frequent and alarming attacks project a significant global loss in the cyber world.

\begin{figure}[ht!]
\begin{center}
\includegraphics[scale=0.44]{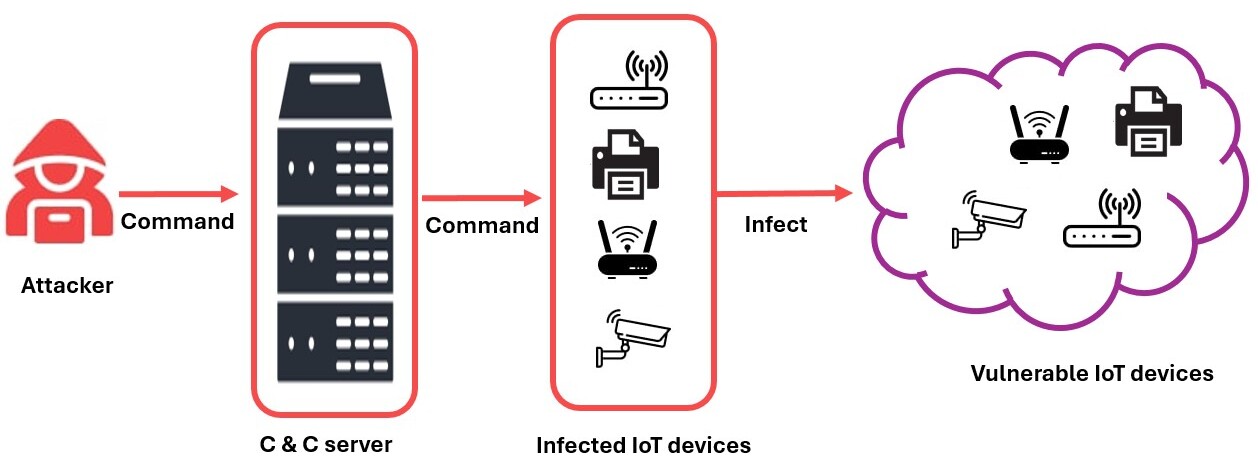}
\caption{Infection spread in IoT devices.}
\label{flow_iot_1}
\end{center}
\end{figure}

It is necessary to understand the malware propagation dynamics in IoT device networks to better control the malware spread. Similar to biological epidemics, the spread of malware in IoT networks can be modeled using the well-established theory of biological epidemics. Specifically, SI (Susceptible $\rightarrow$ Infected), SIR (Susceptible $\rightarrow$ Infected $\rightarrow$ Recovered), and SEIR (Susceptible $\rightarrow$ Exposed $\rightarrow$ Infected $\rightarrow$ Recovered) models are being used to describe malware propagation \cite{cheng2010modeling,wei2019modeling,al2019analysis}. More complex models, by introducing other aspects like the active stage, dormant stage, immune stage, awareness of users, alert mechanism etc., have been studied to capture more specific characteristics \cite{peng2025dynamic,wang2017sadi,aleja2022compartmental,godoi2023spatio}. Several studies have focused on the control strategies of malware propagation and have used Pontryagin's Maximum Principle (PMP) to obtain optimal local control \cite{kazeem2018optimal,hernandez2020optimal, liu2017web}. A detailed literature survey is given in the next section. To the best of our knowledge, we have identified some gaps in this area as follows:

\begin{itemize}
    \item In the context of malware spreading, several key characteristics of epidemic progression, such as the maximum number of infected individuals, the time to reach the maximum infection level, and the region of epidemic growth in the control parameter space, are well appreciated in biological epidemics but have not been studied in depth in malware propagation.

    \item Though there are studies in the existing literature on the optimal control of malware propagation, we found that most of the existing work relies on local optimization techniques such as PMP. However, these local optimization techniques are highly dependent on the initial guess of the control parameters, especially when the cost function is non-convex. In such cases, there is a significant probability that the local optimal control may differ from the global optimal control. Since multi-compartmental epidemic models are highly nonlinear, the corresponding cost functions are often non-convex. Therefore, local optimization methods may fail to achieve the global optimum.
\end{itemize}

In this paper, we focus on addressing the above-mentioned gaps in the theory of malware propagation. The key contributions of this work are as follows:

\begin{itemize}
    \item Propose a generic SEIRV (Susceptible $\rightarrow$ Exposed $\rightarrow$ Infected $\rightarrow$ Recovered $\rightarrow$ Vaccinated) model to describe the characteristics of malware propagation, establish the well-posedness of the model, derive the analytical expression for the malware propagation threshold.

    \item Study the stability of the malware-free equilibrium and the endemic equilibrium, and the existence of forward bifurcation.

     \item Identify the most significant parameters in the model that influence the malware propagation threshold using normalized forward sensitivity indices.

    \item Study the impact of the transmission rate and control parameters on important epidemic characteristics, such as the maximum number of infected, the time to peak infection, and the region of epidemic growth in the control parameter space.

    \item Develop a hybrid gradient-based simulated annealing global optimization algorithm that can be applied to multi-compartmental SEIRV-type models to obtain the global optimal control.

    \item Calibrate the proposed model by fitting the infection counts to the virus propagation data obtained from the “Windows Malware Dataset with PE API Calls” \cite{catak2021data,kovtun2024entropy}.
    
\end{itemize}

The rest of the paper is organized as follows: In Section~\ref{sec2}, we briefly present a literature review on related epidemic models. In Section~\ref{sec3}, we develop the model, establish its positivity and boundedness, derive analytical expressions for the malware propagation threshold and the region of epidemic growth, and study the stability of the malware-free and endemic equilibrium points. In Section~\ref{sec4}, we propose a hybrid gradient-based global optimization method. Model results using numerical simulations are discussed in Section~\ref{sec5}. The proposed model is calibrated using Windows malware dataset in Section~\ref{section_calibration}. Finally, the conclusions are presented in Section~\ref{sec6}.



\section{Epidemic models: From biological to cyber epidemics}\label{sec2}
The foundation of mathematical models to describe the transmission dynamics of biological epidemics was developed by D. Bernoulli in the 18th century and by W. O. Kermack and A. G. McKendrick in the early 20th century \cite{kermack1927contribution,kermack1932contributions,kermack1933contributions}. Following their work, several developments have emerged in epidemic modeling to describe and control biological epidemics \cite{brauer,keeling2008modeling,liu2025flexible,amine2021dynamics}.

At the end of the last century, Cohen and Murray proposed that malware propagation in the cyber world could be analyzed using biological epidemic models \cite{cohen1987computer,murray1988application}. In 1991, Kephart and White first proposed an SIS model (susceptible $\rightarrow$ infected $\rightarrow$ susceptible) to study malware propagation in cyber systems \cite{kephart1992directed}. Following this work \cite{kephart1992directed}, several developments in epidemic modeling for the cyber world have emerged.

Using the concept of multi-compartmental epidemic models, researchers have introduced models for various types of networks, such as IoT networks \cite{li2020dynamics,gardner2017using}, peer-to-peer networks \cite{ramachandran2010dynamics,feng2012modeling}, scale-free networks \cite{hosseini2016model,dadlani2014stability}, and Wi-Fi networks \cite{le2022malware}. A game-theoretic approach to epidemic modeling was explored in \cite{spyridopoulos2015game}. 

A cyber network exhibits heterogeneity across multiple dimensions, such as (a) user-level heterogeneity, (b) device-level heterogeneity, and (c) software-level heterogeneity. The impact of this heterogeneity on malware propagation has been studied by A. Alexeev et al. \cite{alexeev2016malware}. Moreover, advanced cyber network architectures encompass a variety of topologies, and the influence of network topology has been analyzed in \cite{hosseini2014malware}. In recent years, researchers have also begun to integrate machine learning techniques with compartmental epidemic models to better understand the progression of malware \cite{severt2023comparison}. Furthermore, in \cite{mahboubi2020stochastic}, the authors incorporated stochasticity into the compartmental model to account for uncertainty in the spread of the IoT botnet. Several developments of SIR-type multi-compartmental models for specific link type (e.g., infrastructure-based (INF), device-to-device (D2D) etc), and specific malware type (e.g., Internet of Things (IoT) malware, Wireless Sensor Network (WSN) malware) can be found in \cite{chen2022mobility}.

The optimal control of malware propagation is a crucial area in the study of cyber epidemics. Broadly speaking, there are two types of control strategies: vaccination of vulnerable devices and treatment of infected devices. The authors in \cite{kazeem2018optimal} studied the optimal control problem in a deterministic model using Pontryagin’s Maximum Principle (PMP). Moreover, studies such as \cite{hernandez2020optimal, liu2017web} have explored various control strategies using PMP. To the best of our knowledge, most of the existing literature on cyber epidemics use local optimization methods like PMP, which are sensitive to initial guesses and may fail for non-convex cost functions. Motivated by this limitation, in this paper, we aim to develop a global optimization algorithm for compartmental epidemic models using a gradient-based simulated annealing method to achieve more reliable control solutions.


\section{SEIRV model formulation}\label{sec3}

In this section, we develop a multi-compartmental epidemic model using a system of ordinary differential equations (ODEs) to describe the malware propagation in IoT devices. A detailed SEIRV-type model is considered to accurately capture the dynamics of self-propagating malware in IoT devices. The entire population of IoT devices is denoted by $N(t)$, which is not static and varies over time depending on the number of new devices joining the population and the natural death of the devices due to the aging of the devices. The entire population of the devices is then divided into several compartments, such as susceptible ($S(t)$), exposed ($E(t)$), infected ($I(t)$), recovered ($R(t)$) and vaccinated ($V(t)$). Note that we do not consider the natural dead devices in the population count and $N(t)=S(t)+E(t)+I(t)+R(t)+V(t)$ for all times $t \geq 0$. Different compartments of the population are briefly illustrated as follows:
\paragraph{Susceptible ($S(t)$):} The devices that are vulnerable to the malware attack. Devices may be vulnerable to malware attacks due to unpatching, default credentials, or other security flaws.

\paragraph{Exposed ($E(t)$):} The devices that are infected but not yet transmitted the infection to other devices. This compartment is basically takes care of the time delay caused due to waiting period to activate the malware, through command-and-control (C\&C) server.

\paragraph{Infected ($I(t)$):} The devices that are infected and can successfully transmit the infection to other devices. In this compartment the devices are completely compromised with the malware. These devices can be a part of the botnet and can scan the network to spread the malware.

\paragraph{Recovered ($R(t)$):} The devices that have been patched, secured or cleared from the malware. These devices gain immediate immunity; however, they may again become susceptible if users do not protect their devices after reinstallation or software upgrade, due to the evolving nature of the malware.

\paragraph{Vaccinated ($V(t)$):} The devices that are not vulnerable to malware and that have been protected through vaccination at the susceptible stage.

At every real time $t$, the new devices are joining the susceptible group with their default credentials at a rate $\Lambda$, and there is a natural death of devices at a rate $\mu$ due to device-aging. Once a malware successfully enters an IoT network, the susceptible devices get exposed due to interaction with infected devices at a rate $\beta$. Then, the exposed devices will become infectious after a certain time lag at a rate $\alpha$. From the susceptible and exposed compartments, the devices may directly join the recovered compartment at rates $\eta_1$ and $\eta_2$ respectively, if the users reset their default credentials or reset their devices on their own. Here, we consider two standard and widely adopted control strategies: vaccination of susceptible devices and treatment of infected devices, where “vaccination” represents the immunization of vulnerable devices to reduce their susceptibility, and “treatment” corresponds to the disinfection or recovery of already infected devices. The susceptible devices are vaccinated at a rate $c_1$, and they can lose immunity at a rate $\sigma_2$. The infected devices recover through treatment at a rate $c_2$. The recovered devices may lose their immunity and join the susceptible group at a rate $\sigma_1$ if the users do not protect their systems from the evolving nature of the malware. This transmission dynamics is shown in the flow diagram in Figure ~\ref{diagram1} and the governing system of ODEs is given as follows:

\begin{subequations}\label{model1}
    \begin{eqnarray}
    \frac{dS(t)}{dt} &=& \Lambda -\beta S(t) I(t) -\eta_1 S(t) +\sigma_1 R(t) + \sigma_2 V(t)-c_1 S(t)-\mu S(t),\\
    \frac{dE(t)}{dt} &=& \beta S(t) I(t) -\alpha E(t) -\eta_2 E(t)-\mu E(t),\\
    \frac{dI(t)}{dt} &=& \alpha E(t) -c_2 I(t)-\mu I(t),\\
    \frac{dR(t)}{dt} &=& \eta_1 S(t)+ \eta_2 E(t)+c_2 I(t) -\sigma_1 R(t)-\mu R(t),\\
    \frac{dV(t)}{dt} &=& c_1 S(t) -\sigma_2 V(t)-\mu V(t),
    \end{eqnarray}
\end{subequations}
with the initial conditions: $S(0)=S_0\geq 0,$ $E(0)=E_0\geq 0,$ $I(0)=I_0\geq 0,$ $R(0)=R_0\geq 0,$ and $V(0)=V_0\geq 0.$

All the parameter values involved in the model \eqref{model1} are assumed to be non-negative. The detailed descriptions of the parameters are listed in Table~\ref{tab:parameters}. It is noteworthy to mention that the parameter values in Table~\ref{tab:parameters} are for illustration purposes, since the objective of this research is to explore the qualitative nature of malware propagation and to develop methodologies that will aid future studies in this direction. However, the parameter values are taken from existing literature as default values, as referenced in Table~\ref{tab:parameters}.

\begin{figure}[ht!]
\begin{center}
	\includegraphics[scale=0.85]{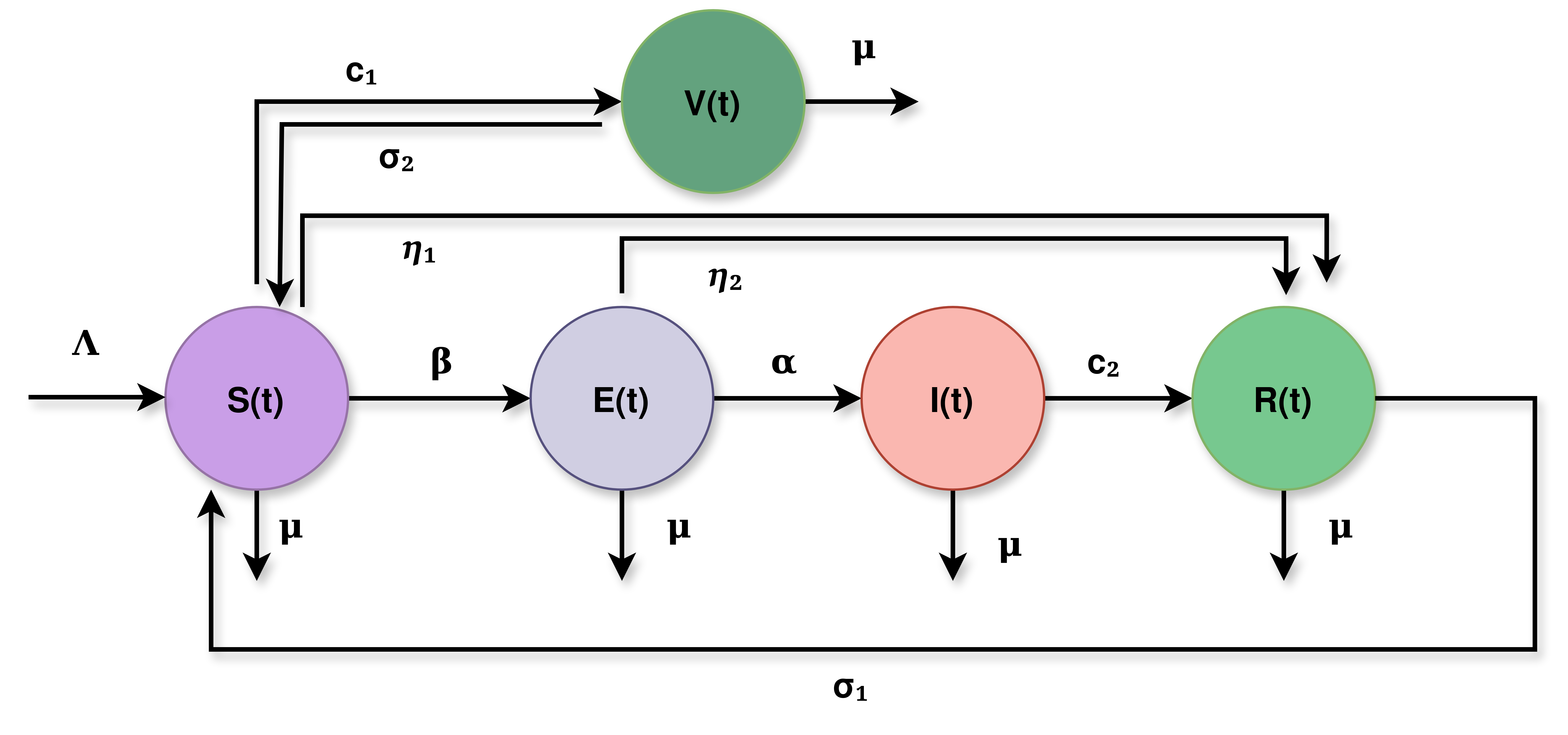}
\caption{Schematic diagram for the model \eqref{model1}.}
\label{diagram1}
\end{center}
\end{figure}

\subsection{Positivity and boundedness}
The total device population must remain nonnegative and bounded at all times. Therefore, it is essential to ensure that the governing system \eqref{model1} satisfies these two properties.
Using the approach in \cite{yang1996permanence}, we prove the positivity of the system \eqref{model1}. Set, 
$
    X(t) = (S(t), E(t), I(t),  R(t), V(t)) \in \mathbb{R}_{0+}^5.
$
Then
\begin{subequations}\scriptsize
 \begin{eqnarray*} 
\left.\frac{dS(t)}{dt}\right|_{S(t) = 0, X(t) \in \mathbb{R}_{0^+}^5}&=& \Lambda +\sigma_1 R(t) +\sigma_2 V(t) \geq 0, \;\;\;\;
\left.\frac{dE(t)}{dt}\right|_{E(t) = 0, X(t) \in \mathbb{R}_{0^+}^5}= \beta S(t) I(t) \geq 0, \\
\left.\frac{dI(t)}{dt}\right|_{I(t) = 0, X(t) \in \mathbb{R}_{0^+}^5}&=& \alpha E(t) \geq 0, \;\;\;\;
\left.\frac{dR(t)}{dt}\right|_{R(t) = 0, X(t) \in \mathbb{R}_{0^+}^5}= \eta_1 S(t) +\eta_2 E(t)+ c_2 I(t) \geq 0, \\
\left.\frac{dV(t)}{dt}\right|_{V(t) = 0, X(t) \in \mathbb{R}_{0^+}^5}&=& c_1 S(t) \geq 0. \\
\end{eqnarray*}
\end{subequations}
Consequently, it proves that $
\mathbb{R}_{0^+}^5$ is invariant under the system \eqref{model1}. To prove the boundedness of the solution, we add all the equations in \eqref{model1} and obtain:
\begin{subequations}
 \begin{eqnarray*}
\frac{dN(t)}{dt}&=& \Lambda -\mu N(t), 
 \end{eqnarray*}
 \end{subequations}
 where, $N(t) =S(t)+E(t)+I(t)+R(t)+V(t).$ This implies 
$$N(t)=\frac{\Lambda}{\mu} + \left(N(0)-\frac{\Lambda}{\mu} \right) e^{-\mu t} \leq \max \left\{ N(0), \frac{\Lambda}{\mu}\right\} =: \widetilde{N}, \;\;\;\; \forall t >0.$$
This shows that at every time point $t$ the total device population is bounded.

\begin{table}[!ht]
    \centering
    \begin{tabular}{|l|p{6cm}|p{3cm}|c|}
        \hline
        \textbf{Parameter} & \textbf{Description} & \textbf{Value/Range} & \textbf{Source} \\
        \hline
        $\beta$ & Malware transmission rate & $4 \times 10^{-9}$ & \cite{gardner2017using} \\
        \hline
         $\alpha$ & $E$ to $I$ transmission rate & 0.25 & \cite{gardner2017using} \\
        \hline
        $\Lambda$ & Rate of new incoming devices &  $0.2292 \times 10^6$ & \cite{gardner2017using}\\
        \hline
         $\eta_1$ & Average movement rate from $S$ to $R$ during botnet growth and decay phase&  $0.10415$ & \cite{gardner2017using}\\
        \hline
        $\eta_2$ & Average movement rate from $E$ to $R$ during botnet growth and decay phase&  $0.10415$ & \cite{gardner2017using}\\
        \hline
         $\sigma_1$ & Rate of relapse from $R$ to $S$&  $0.00417$ & \cite{gardner2017using}\\
        \hline
         $\sigma_2$ & Movement rate from $V$ to $S$ &  $0.00417$ & Assumption\\
        \hline
        $c_1$ & Rate of vaccination of susceptible devices &  $(0,1)$ & Assumption\\
        \hline
        $c_2$ & Rate of treatment of infected devices &  $(0,1)$ & Assumption\\
        \hline
         $\mu$ & Natural death rate of devices &  $0.0004$ & \cite{ngoy2021combining}\\
        \hline
    \end{tabular}
    \caption{Details of the model parameters along with their corresponding values.}
    \label{tab:parameters}
\end{table}

\subsection{Malware-free equilibrium}
In this subsection we derive the malware-free steady state of the system \eqref{model1} and determine its stability. In the context of cyber epidemics, the malware-free Equilibrium (MFE) represents a stable network state where no devices are compromised or under attack — that is, there are no actively infected systems present. This scenario can correspond to a situation where either the malware has been successfully eradicated, or preventive measures (such as firewalls, patches, and antivirus systems) have effectively suppressed the outbreak. We find the malware-free Equilibrium (MFE) by setting derivatives to zero and solving the following system of algebraic equations:

\begin{subequations}\label{equilibrium_equn1}
    \begin{eqnarray}
     &\Lambda -\beta S(t) I(t) -\eta_1 S(t) +\sigma_1 R(t) + \sigma_2 V(t)-c_1 S(t)-\mu S(t) = 0,\label{equilibrium_equn1a}\\
     &\beta S(t) I(t) -\alpha E(t) -\eta_2 E(t)-\mu E(t)= 0,\label{equilibrium_equn1b}\\
     &\alpha E(t) -c_2 I(t)-\mu I(t)= 0,\label{equilibrium_equn1c}\\
     &\eta_1 S(t)+ \eta_2 E(t) +c_2 I(t) -\sigma_1 R(t)-\mu R(t)= 0,\label{equilibrium_equn1d}\\
     &c_1 S(t) -\sigma_2 V(t)-\mu V(t)= 0.\label{equilibrium_equn1e}
    \end{eqnarray}
\end{subequations}


At the malware-free equilibrium (MFE), we assume that there are no exposed or infected individuals, i.e., $
E = I = 0.
$
Then solving the system \eqref{equilibrium_equn1} we obtain the MFE $\mathcal{E}^0 = (S^0, E^0, I^0, R^0, V^0)$, given by:
\[
\begin{aligned}
S^0 &= \frac{\Lambda}{\eta_1-\frac{\sigma_1 \eta_1}{\sigma_1 +\mu} + c_1 +\mu -\frac{c_1 \sigma_2}{\sigma_2 +\mu}}, \\
E^0 &= 0,\;\;
I^0 = 0, \\
R^0 &= \frac{\eta_1 S^0}{\sigma_1 + \mu},\;\;
V^0 = \frac{c_1 S^0}{\sigma_2 + \mu}.
\end{aligned}
\]
For the feasibility of the MFE $\mathcal{E}^0$, all the components of $\mathcal{E}^0$ should be non-negative. It is clear that $E^0$, $I^0$, $R^0$ and $V^0$ are non-negative for any positive parameter values. Moreover, the denominator of $S^0$ can be written as $\eta_1-\frac{\sigma_1 \eta_1}{\sigma_1 +\mu} + c_1 +\mu -\frac{c_1 \sigma_2}{\sigma_2 +\mu}= \frac{\eta_1 \mu}{\sigma_1 +\mu} +c_1 +\frac{c_1 \mu}{\sigma_2 +\mu}$, which proves the non-negativity of the MFE.

\subsection{Malware transmission threshold and the stability of MFE} 
The malware transmission threshold ($\mathcal{R}_c$) can be defined as the number of secondary infected devices caused by a single infected device in a completely susceptible network. This threshold is important because it indicates whether malware propagation will grow or decay. To derive the expression of $\mathcal{R}_c$ Following the next-generation matrix approach described in \cite{BRN}, we determine the malware transmission threshold for system (\ref{model1}). In this formulation, the compartments $S$, $R$, and $V$ are classified as uninfected classes, while $E$ and $I$ represent the infected classes. The matrix $\mathcal{F}$ corresponds to the rate of appearance of new infections, whereas the matrix $\mathcal{V}$ represents the transfer terms accounting for transitions out of the infected compartments. These matrices are given by:
$$\mathcal{F} = \left(
\begin{array}{c}
\mathcal{F}_1 \\
\mathcal{F}_2
\end{array}
\right)=\left(
\begin{array}{c}
\beta S I \\
\alpha E
\end{array}
\right),\;\;\;\;
 \mathcal{V} =\left(
\begin{array}{c}
\mathcal{V}_1 \\
\mathcal{V}_2
\end{array}
\right)= \left(
\begin{array}{c}
\alpha E +\eta_2 E + \mu E\\
c_2 I+\mu I
\end{array}
\right).$$

The Jacobian matrices corresponding to $\mathcal{F}$ and $\mathcal{V}$ are obtained by evaluating them at the malware-free equilibrium $\mathcal{E}^0$ and are respectively given by:
$$
F = \left(
\begin{array}{cc}
\frac{\partial \mathcal{F}_1}{\partial E} & \frac{\partial \mathcal{F}_1}{\partial I} \\[0.3em]
\frac{\partial \mathcal{F}_2}{\partial E}  & \frac{\partial \mathcal{F}_2}{\partial I}  
\end{array}
\right)\bigg|_{\text{at}\ \mathcal{E}^0}
=
\left(
\begin{array}{cc}
0 & \beta S^0 \\
\alpha & 0 
\end{array}
\right),
$$

and
$$
V = \left(
\begin{array}{cc}
\frac{\partial \mathcal{V}_1}{\partial E} & \frac{\partial \mathcal{V}_1}{\partial I} \\[0.3em]
\frac{\partial \mathcal{V}_2}{\partial E}  & \frac{\partial \mathcal{V}_2}{\partial I}  
\end{array}
\right)\bigg|_{\text{at}\ \mathcal{E}^0}
=
\left(
\begin{array}{cc}
\alpha + \eta_2 + \mu & 0 \\[0.3em]
0 & c_2 + \mu
\end{array}
\right).
$$
Then
$$
FV^{-1} = \left(
\begin{array}{cc}
0 & \frac{\beta S^0}{c_2+\mu} \\[0.3em]
\frac{\alpha}{\alpha+\eta_2+\mu} & 0
\end{array}
\right).
$$
The spectral radius of $FV^{-1}$ determines the malware transmission threshold $\mathcal{R}_c$, given by:
$$
\mathcal{R}_c = \sqrt{ \frac{\beta S^0 \alpha}{(c_2 + \mu)(\alpha + \eta_2 + \mu)}}.
$$
The stability of the MFE is characterized by the malware transmission threshold $\mathcal{R}_c$ and we have the following theorem:

\begin{theorem}
    The MFE $\mathcal{E}^0$ is locally asymptotically stable if $\mathcal{R}_c<1$ and unstable if $\mathcal{R}_c>1$.
\end{theorem}

\begin{proof}
    To examine the local asymptotic stability of the MFE, we derive the Jacobian matrix of the whole system \eqref{model1} around the MFE $\mathcal{E}^0$ given by:

\begin{center}
    $J =
\begin{pmatrix}
-(\eta_1+c_1+\mu) & 0 & -\beta S^0 & \sigma_1 & \sigma_2\\[4pt]
0 & -(\alpha+\eta_2+\mu) & \beta S^0 & 0 & 0\\[4pt]
0 & \alpha & -(c_2+\mu) & 0 & 0\\[4pt]
\eta_1 & \eta_2 & c_2 & -(\sigma_1+\mu) & 0\\[4pt]
c_1 & 0 & 0 & 0 & -(\sigma_2+\mu)
\end{pmatrix}.$
\end{center}
The characteristic equation for $J$  is given by:

\begin{eqnarray}\label{char_eqn_1}
(\lambda+\mu)\,
\big(\lambda^2 + L_1\lambda + L_2\big)\,
\big(\lambda^2 + L_3\lambda + L_4\big)=0,
\end{eqnarray}

where,
\begin{subequations}
\begin{eqnarray*}
   L_1 &=& \alpha + c_2 + \eta_2 + 2\mu\; >0,\\
L_2 &=& (c_2+\mu)(\alpha+\eta_2+\mu) - \alpha\beta S^0= (c_2+\mu)(\alpha+\eta_2+\mu) (1-\mathcal{R}_c^2),\\
L_3 &=& c_1 + \eta_1 + \sigma_1 + \sigma_2 + 2\mu\; >0,\\
L_4 &=& c_1\mu + c_1\sigma_1 + \eta_1\mu + \eta_1\sigma_2 + \mu^2 + \mu\sigma_1 + \mu\sigma_2 + \sigma_1\sigma_2\; >0.
\end{eqnarray*}
\end{subequations}
The eigen values of $J$ can be determined by the solution of the equation \eqref{char_eqn_1} and are given as follows:
$$\lambda_1=-\mu <0,\;\; \lambda_2=\frac{-L_1 + \sqrt{L_1^2 -4L_2}}{2},\;\;\lambda_3=\frac{-L_1 - \sqrt{L_1^2 -4L_2}}{2},$$
$$\lambda_4=\frac{-L_3 + \sqrt{L_3^2 -4L_4}}{2},\;\text{and}\;\lambda_5=\frac{-L_3 - \sqrt{L_3^2 -4L_4}}{2}.$$
It is obvious that $\lambda_1$, $\lambda_4$, and $\lambda_5$ are always negative, and the signs of $\lambda_2$ and $\lambda_3$ are determined by $\mathcal{R}_c$. Note that, by Descartes’ rule of signs, $\lambda_2$ and $\lambda_3$ cannot be complex conjugates. Suppose $\mathcal{R}_c < 1$. Then $L_2 > 0$, and consequently $\lambda_2 < 0$ and $\lambda_3 < 0$. If $\mathcal{R}_c > 1$, then $L_2 < 0$, and consequently $\lambda_2 > 0$ and $\lambda_3 < 0$. This completes the proof.
\end{proof}

\begin{theorem}\label{thm:DFE_GAS}

If \(\mathcal{R}_0<1\), then the MFE $\mathcal{E}^0$ is globally asymptotically stable in the feasible region
\(\{(S,E,I,R,V)\in\mathbb{R}_{\ge0}^5\}\). 
\end{theorem}

\begin{proof}
To prove the global stability we follow the approach as described in \cite{castillo2002computation}. Consider the infected subsystem for \((E,I)\):
\begin{align}
\dot E &= \beta S(t) I -(\alpha+\eta_2+\mu)E,\label{eq:E}\\
\dot I &= \alpha E - (c_2+\mu)I.\label{eq:I}
\end{align}
Since \(S(t)\le N(t)\le \widetilde{N}\) for all \(t\), we can compare
\((E,I)\) with the following linear comparison system
\[
\frac{d}{dt}\begin{pmatrix}\widehat{E}\\[4pt]\widehat{I}\end{pmatrix}
= A \begin{pmatrix}\widehat{E}\\[4pt]\widehat{I}\end{pmatrix},\qquad
A=\begin{pmatrix}
-(\alpha+\eta_2+\mu) & \beta \widetilde{N}\\[4pt]
\alpha & -(c_2+\mu)
\end{pmatrix},
\]
with the same nonnegative initial data. The system is cooperative, so by
the comparison principle \(0\le (E(t),I(t))\le(\widehat{E}(t),\widehat{I}(t))\)
for all \(t\ge0\), where \((\widehat{E},\widehat{I})\) solves the linear
system above.

The linear matrix \(A\) has trace
\(\operatorname{tr}(A)=-(\alpha+\eta_2+\mu)-(c_2+\mu)<0\) and determinant
\[
\det(A)=(\alpha+\eta_2+\mu)(c_2+\mu)-\alpha\beta \widetilde{N}.
\]
Observing that \(S^0\le \widetilde{N}\) and using the definition of
\(\mathcal{R}_0\) we have, 
\[
\mathcal{R}_0<1
\quad\Longrightarrow\quad
\frac{\beta\alpha \widetilde{N}}{(\alpha+\eta_2+\mu)(c_2+\mu)}<1,
\]
hence \(\det(A)>0\). Therefore both eigenvalues of \(A\) are real and
negative, and the linear system is exponentially stable, i.e., there exist
constants \(C,\lambda>0\) such that
\(\|(\widehat{E}(t),\widehat{I}(t))\|\le C e^{-\lambda t}\|(\widehat{E}(0),\widehat{I}(0))\|\).
By comparison it follows that \(E(t)\to0\) and \(I(t)\to0\) exponentially
as \(t\to\infty\).

Now, consider the subsystem \((S,R,V)\) putting $(E, I)=(0, 0)$. From system
\eqref{model1}, we have
\[
\begin{pmatrix}
\dot S\\[4pt]
\dot R\\[4pt]
\dot V
\end{pmatrix}
=
\begin{pmatrix}
-(\eta_1+c_1+\mu) & \sigma_1 & \sigma_2\\[4pt]
\eta_1 & -(\sigma_1+\mu) & 0\\[4pt]
c_1 & 0 & -(\sigma_2+\mu)
\end{pmatrix}
\begin{pmatrix}
S\\[4pt]
R\\[4pt]
V
\end{pmatrix}
+
\begin{pmatrix}
\Lambda\\[4pt]
0\\[4pt]
0
\end{pmatrix}.
\]
Consider
$$B=\begin{pmatrix}
-(\eta_1+c_1+\mu) & \sigma_1 & \sigma_2\\[4pt]
\eta_1 & -(\sigma_1+\mu) & 0\\[4pt]
c_1 & 0 & -(\sigma_2+\mu).
\end{pmatrix}$$

The characteristic equation for the matrix $B$ is given by:

$$\lambda^3+L_1 \lambda^2 +L_2 \lambda +L_3=0,$$

where,
$$L_1=a+b+d,$$
$$L_2=ab+ad+bd-\sigma_1 \eta_1 -\sigma_2 c_1,$$
$$L_3= abd-\sigma_1 \eta_1 d - \sigma_2 c_1 b,$$
$$a=\eta_1 +c_1 +\mu, b= \sigma_1 +\mu, d= \sigma_2 +\mu.$$

Clearly, $L_1>0$. 
Now,
\begin{eqnarray*}
    L_2&=& \left[ (\eta_1 +c_1 +\mu)(\sigma_1 +\mu) -\sigma_1 \eta_1\right] + bd + \left[ (\eta_1 +c_1 +\mu)(\sigma_2 +\mu) \right]\\
    &=& \left[ \eta_1 \mu +(c_1 +\mu)(\sigma_1 +\mu)\right] + bd + \left[ (\eta_1 +\mu)(\sigma_2 +\mu) +c_1 \mu \right] >0.
\end{eqnarray*}
Next,
\begin{eqnarray*}
    L_3&=& abd-\sigma_1 \eta_1 d - \sigma_2 c_1 b = d(ab-\sigma_1 \eta_1)-\sigma_2 c_1 b\\
    &=& d \left[ (\eta_1 +c_1 +\mu)(\sigma_1 +\mu) -\sigma_1 \eta_1 \right]-\sigma_2 c_1 b\\
    &=& d \left[ \eta_1 \mu  +(c_1 +\mu)b  \right]-\sigma_2 c_1 b\\
    &=& d \eta_1 \mu +b \left[ (c_1 +\mu)(\sigma_2 +\mu) -\sigma_2 c_1\right]\\
    &=& d \eta_1 \mu +b \left[ c_1 \mu  +\mu(\sigma_2 +\mu)\right] >0.
\end{eqnarray*}

After some algebraic simplifications, $L_1 L_2-L_3$ can be written as follows: 

\begin{eqnarray*}
    L_1L_2-L_3&=& a^2 b +ab^2 +a^2d +ad^2+ b^2 d +b d^2 + 2abd -\sigma_1 \eta_1 (a+b) -\sigma_2 c_1 (a+d)\\
    &=& (a+b) \left[ ab -\sigma_1 \eta_1\right] +(a+d) \left[ ad -\sigma_2 c_1\right]+ b^2 d +b d^2 + 2abd\\
    &=& (a+b) \left[ (\eta_1 +c_1 +\mu)(\sigma_1 +\mu) -\sigma_1 \eta_1\right] \\
    &&+(a+d) \left[ (\eta_1 +c_1 +\mu)(\sigma_2 +\mu) -\sigma_2 c_1\right]+ b^2 d +b d^2 + 2abd\\
    &=& (a+b) \left[ \eta_1 \mu  +(c_1 +\mu)(\sigma_1 +\mu)\right] \\
    &&+ (a+d) \left[ (\eta_1 +\mu)(\sigma_2 +\mu) +c_1 \mu \right]+ b^2 d +b d^2 + 2abd >0.
\end{eqnarray*}
Thus by the Routh-Hurwitz conditions, the matrix $B$ has all eigen values with negative real parts. Consequently, $(S^0, R^0, V^0)$ is globally asymptotically stable for the $(S, R, V)$ sub-system as defined above. Thus, using the theorem discussed in \cite{castillo2002computation}, we can conclude that the DFE $\mathcal{E}^0$ is globally asymptotically stable for the system \eqref{model1} if $\mathcal{R}_0<1$. This completes the proof.

\end{proof}

\subsection{Existence of endemic equilibrium}
We have the following theorem for the existence of endemic equilibrium.

\begin{theorem}
    The system \eqref{model1} has unique endemic equilibrium point \((S^e,E^e,I^e,R^e, V^e)\) if and only if $\mathcal{R}_0 >1$.
\end{theorem}
\begin{proof}

Let us denote the endemic equilibrium (EE) point by \((S^e,R^e,E^e,I^e,V^e)\). 

From the equation of \(I\) in the system \eqref{model1} at EE
\[
0=\alpha E - (c_2+\mu)I \quad\Longrightarrow\quad
E^e=\frac{c_2+\mu}{\alpha}\,I^e.
\]

Using the above equation and the equation of \(E\) in the system \eqref{model1} at EE

\[
0=\beta S -(\alpha+\eta_2+\mu)\frac{c_2+\mu}{\alpha}
\quad\Longrightarrow\quad
\,S^e=\dfrac{(\alpha+\eta_2+\mu)(c_2+\mu)}{\beta\alpha}\,.
\]

From the equation of \(V\) in the system \eqref{model1} at EE
\[
0=c_1 S -(\sigma_2+\mu)V \quad\Longrightarrow\quad
\,V^e=\dfrac{c_1}{\sigma_2+\mu}S^e.
\]

From the equation of \(R\) in the system \eqref{model1} at EE
\[
0=\eta_1 S +\eta_2 E + c_2 I -(\sigma_1+\mu)R
\quad\Longrightarrow\quad
\,R^e=\dfrac{\eta_1 S^e +\eta_2 E^e + c_2 I^e}{\sigma_1+\mu}\,.
\]

Finally substitute the above relations into the \(S\) equation at EE, we get

\[
\begin{aligned}
0 &= \Lambda -\beta S^e I^e -(\eta_1+c_1+\mu)S^e
+ \sigma_1\frac{\eta_1 S^e +\eta_2 E^e + c_2 I^e}{\sigma_1+\mu}
+ \sigma_2\frac{c_1}{\sigma_2+\mu}S^e \\
&= A_0
\;-\;A_1\,I^e,
\end{aligned}
\]
where,
$$A_0=\Lambda -(\eta_1+c_1+\mu)S^e
+ \sigma_1\frac{\eta_1 S^e}{\sigma_1+\mu}
+ \sigma_2\frac{c_1 S^e}{\sigma_2+\mu},$$
$$A_1=\Big(\beta S^e -\frac{\sigma_1 c_2}{\sigma_1+\mu}-\frac{\sigma_1 \eta_2}{\sigma_1+\mu} \frac{c_2+\mu}{\alpha}\Big).$$
Thus
\[
I^e=\frac{A_0}{A_1}=\frac{\Lambda -(\eta_1+c_1+\mu)S^e
+ \sigma_1\frac{\eta_1 S^e}{\sigma_1+\mu} + \sigma_2\frac{c_1 S^e}{\sigma_2+\mu}}
{\beta S^e -\frac{\sigma_1 c_2}{\sigma_1+\mu}-\frac{\sigma_1 \eta_2}{\sigma_1+\mu} \frac{c_2+\mu}{\alpha}}.
\]
Note that $S^0=\Lambda/D,$ where, $D=(\eta_1+c_1+\mu)
- \sigma_1\frac{\eta_1}{\sigma_1+\mu} -\sigma_2\frac{c_1}{\sigma_2+\mu}.$

Now, $$A_0=\Lambda -S^eD=D(S^0-S^e).$$
Moreover note that $S^e=\frac{S^0}{\mathcal{R}_0^2}.$

This implies $$A_0=\frac{D S^0}{\mathcal{R}_0^2} (\mathcal{R}_0^2-1).$$

Thus, $A_0>0$ if and only if $\mathcal{R}_0>1.$

Moreover, $$\alpha (\sigma_1+\mu)A_1=\mu \big( \sigma_1 (\alpha +c_2+\mu) + (\alpha+\eta_2 +\mu) (c_2 +\mu) \big),$$
which ensures that $A_1>0$. 
Thus an endemic equilibrium exists uniquely if and only if
\(\mathcal{R}_0>1\). 
\end{proof}
\subsection{Local stability of the endemic equilibrium}
To study the local asymptotical stability of the endemic equilibrium point \((S^e,E^e,I^e,R^e, V^e)\), we calculate the Jacobian matrix of the model \eqref{model1} around the endemic equilibrium and is give as follows:

\begin{center}
    $J =
\begin{pmatrix}
-(\omega_1+\beta I^e) & 0 & -\beta S^e & \sigma_1 & \sigma_2\\[4pt]
\beta I^e & -\omega_2 & \beta S^e & 0 & 0\\[4pt]
0 & \alpha & -\omega_3 & 0 & 0\\[4pt]
\eta_1 & \eta_2 & c_2 & -\omega_4 & 0\\[4pt]
c_1 & 0 & 0 & 0 & -\omega_5
\end{pmatrix},$
\end{center}
where,
\begin{eqnarray*}
    \omega_1&=& \eta_1+c_1+\mu,\;\;
    \omega_2= \alpha +\eta_2+\mu,\\
    \omega_3 &=& c_2 +\mu,\;\; \omega_4 = \sigma_1 +\mu,\;\; \omega_5 = \sigma_2+\mu.
\end{eqnarray*}
The characteristic equation for $J$  is given by:

\begin{eqnarray}\label{char_eqn_2}
\lambda^5 + H_1\lambda^4 + H_2 \lambda^3 + H_3\lambda^2 + H_4 \lambda +H_5=0,
\end{eqnarray}
where
\[
\begin{aligned}
H_1 &= \omega_5 + \omega_4 - \omega_3 + \omega_2 + \omega_1 + b I^{e}, \\[4pt]
H_2 &= -\bigl(
       c_1 c_2
     - \omega_5\omega_4 + \omega_5\omega_3 - \omega_5\omega_2 - \omega_5\omega_1 - \omega_5 b I^{e}
     + n_1 c_1
     + \omega_4\omega_3 - \omega_4\omega_2 - \omega_4\omega_1 - \omega_4 b I^{e} \\[-2pt]
    &\quad + a b S^{e}
     + \omega_3\omega_2 + \omega_3\omega_1 + \omega_3 b I^{e}
     - \omega_2\omega_1 - \omega_2 b I^{e}
     \bigr), \\[6pt]
H_3 &= -\bigl(
       c_1 c_2 \omega_4 - c_1 c_2 \omega_3 + c_1 c_2 \omega_2
     + \omega_5 n_1 c_1
     + \omega_5\omega_4\omega_3 - \omega_5\omega_4\omega_2 - \omega_5\omega_4\omega_1 - \omega_5\omega_4 b I^{e} \\
    &\quad + \omega_5 a b S^{e}
     + \omega_5\omega_3\omega_2 + \omega_5\omega_3\omega_1 + \omega_5\omega_3 b I^{e}
     - \omega_5\omega_2\omega_1 - \omega_5\omega_2 b I^{e} \\[-2pt]
    &\quad - n_1 c_1 \omega_3 + n_1 c_1 \omega_2
     + \omega_4 a b S^{e}
     + \omega_4\omega_3\omega_2 + \omega_4\omega_3\omega_1 + \omega_4\omega_3 b I^{e}
     - \omega_4\omega_2\omega_1 - \omega_4\omega_2 b I^{e} \\[-2pt]
    &\quad + a b S^{e}\,\omega_1 + \omega_3\omega_2\omega_1 + \omega_3\omega_2 b I^{e} + n_2 b I^{e} c_1
     \bigr), \\[6pt]
H_4 &= -\bigl(
      -c_1 c_2 a b S^{e} - c_1 c_2 \omega_3 \omega_2
      + \omega_5 n_2 b I^{e} c_1 - \omega_5 n_1 c_1 \omega_3 + \omega_5 n_1 c_1 \omega_2 \\
    &\quad + \omega_5 \omega_4 a b S^{e}
      + \omega_5 \omega_4 \omega_3 \omega_2 + \omega_5 \omega_4 \omega_3 \omega_1 + \omega_5 \omega_4 \omega_3 b I^{e}
      - \omega_5 \omega_4 \omega_2 \omega_1 - \omega_5 \omega_4 \omega_2 b I^{e} \\
    &\quad + \omega_5 a b S^{e}\,\omega_1 + \omega_5 \omega_3 \omega_2 \omega_1 + \omega_5 \omega_3 \omega_2 b I^{e}
      - c_1 c_2 \omega_4 \omega_3 + c_1 c_2 \omega_4 \omega_2 \\
    &\quad + c_2 a b I^{e} c_1 - n_2 b I^{e} c_1 \omega_3 - n_1 c_1 a b S^{e} - n_1 c_1 \omega_3 \omega_2
      + \omega_4 a b S^{e}\,\omega_1 + \omega_4 \omega_3 \omega_2 \omega_1 + \omega_4 \omega_3 \omega_2 b I^{e}
      \bigr), \\[6pt]
H_5 &= c_1 c_2 \omega_4 a b S^{e} + c_1 c_2 \omega_4 \omega_3 \omega_2
     - \omega_5 c_2 a b I^{e} c_1 + \omega_5 n_2 b I^{e} c_1 \omega_3 \\
    &\quad + \omega_5 n_1 c_1 a b S^{e} + \omega_5 n_1 c_1 \omega_3 \omega_2
     - \omega_5 \omega_4 a b S^{e}\,\omega_1 - \omega_5 \omega_4 \omega_3 \omega_2 \omega_1 - \omega_5 \omega_4 \omega_3 \omega_2 b I^{e}.
\end{aligned}
\]

All roots of the characteristic polynomial \eqref{char_eqn_2} have negative real parts if and only if the following Routh–Hurwitz conditions are satisfied:

\begin{enumerate}
\item $H_1 > 0,  H_2 > 0,  H_3 > 0,  H_4 > 0,  H_5 > 0,$

\item $H_1 H_2 - H_3 > 0,$
\item $H_1 H_4 - H_5 > 0,$
\item $(H_1 H_2 - H_3)H_3 - H_1(H_1 H_4 - H_5) > 0,$
\item $\big[\big( (H_1 H_2 - H_3)H_3 - H_1(H_1 H_4 - H_5) \big)(H_1 H_4 - H_5) - (H_1 H_2 - H_3)H_5 \big] > 0.$
\end{enumerate}

 Due to the complexity of the expressions appeared in the characteristic equation of the Jacobian corresponding to the endemic equilibrium point, further analytical investigation was not possible and we demonstrate the results with numerical simulation. Using the parameter setup as mentioned in Table \ref{tab:parameters} ad $c_1=0.05, c_2=0.05$, the characteristic equation \eqref{char_eqn_2} becomes:

\[
\lambda^5 + 14.704\, \lambda^4 + 5.192\, \lambda^3 + 0.170\, \lambda^2 + 0.00226\, \lambda + 5.02\times 10^{-7} = 0,
\]
and the eigen values are given by
\[
 \lambda = \{-14.295,\ -0.3546,\ -0.0504,\ -0.00457,\ -0.00457\}.
\]

This guarantees that the endemic equilibrium point is locally stable under the above parameter setup. Moreover, numerically we observe the existence of forward bifurcation in the model \eqref{model1} ( see Figure \ref{forward_bifurcation}). 

\begin{figure}[ht!]
\begin{center}
	\includegraphics[scale=0.6]{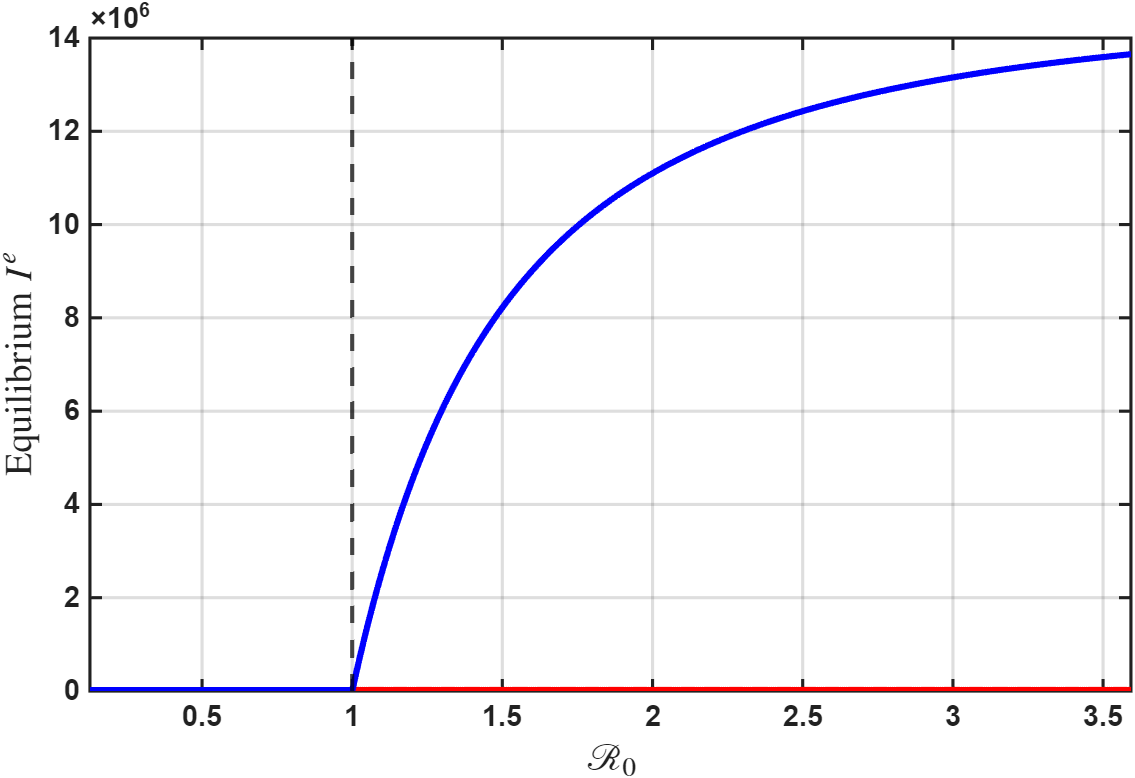}
\caption{Existence of forward bifurcation in the model \eqref{model1} under the parameter setup as mentioned in Table \eqref{tab:parameters}. Blue color corresponds to stable equilibrium point and red color corresponds to unstable equilibrium point.}
\label{forward_bifurcation}
\end{center}
\end{figure}

\color{black}
\subsection{Malware growth and decay regions in control space}
The malware transmission threshold $\mathcal{R}_c$ is an important threshold that can classify the malware growth or decay regions. Malware growth occurs if $\mathcal{R}_c>1$ and there is no malware growth if $\mathcal{R}_c <1$. Since we have two control parameters  $c_1$ and $c_2$, the $(c_1, c_2)$-parametric space can be divided into these regions. The separatrix between malware extinction and malware growth regions is characterized by $\mathcal{R}_c=1$, and the regions are defined respectively as follows:
$$\mathcal{S}_1 :=  \left\{ (c_1, c_2) \in [0, 1] \times [0, 1]: \beta S^0 \alpha < (c_2 + \mu) (\alpha +\eta_2 +\mu) \right\},$$
and
$$\mathcal{S}_2 :=  \left\{ (c_1, c_2) \in [0, 1] \times [0, 1]: \beta S^0 \alpha > (c_2 + \mu) (\alpha +\eta_2 +\mu) \right\}.$$
Figure \ref{growth_region_plot} shows the regions for different values of $\beta$. The dark color regions represent malware growth regions and light color regions represent the malware extinction regions. As the value of transmission rate $\beta$ increases, the area of the epidemic growth region increases in the $(c_1, c_2)$-parametric plane.

\begin{figure}[ht!]
\begin{center}
		\mbox{
  \subfigure[]{\includegraphics[width=0.45\textwidth,height=0.35\textwidth]{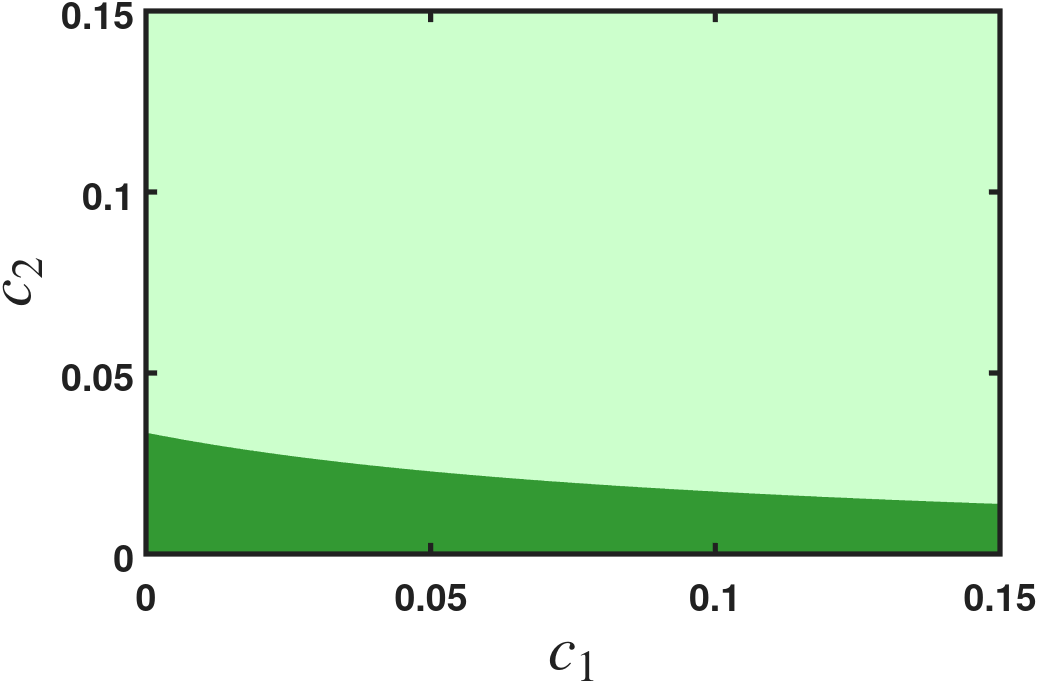}}
  \subfigure[]{\includegraphics[width=0.45\textwidth,height=0.35\textwidth]{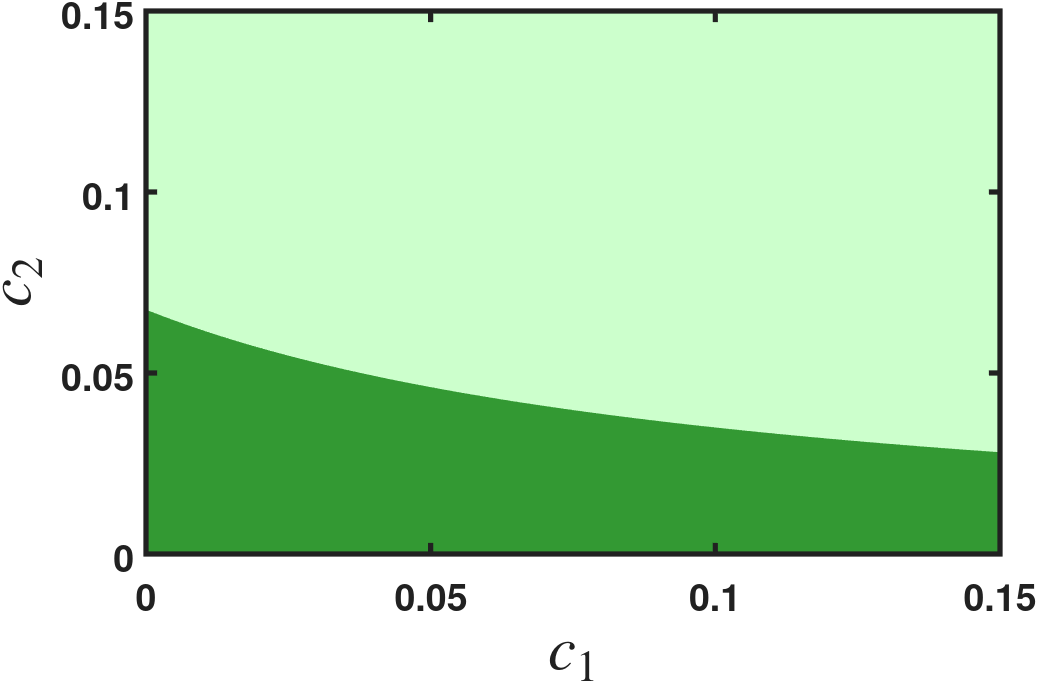}}}
  \mbox{
  \subfigure[]{\includegraphics[width=0.45\textwidth,height=0.35\textwidth]{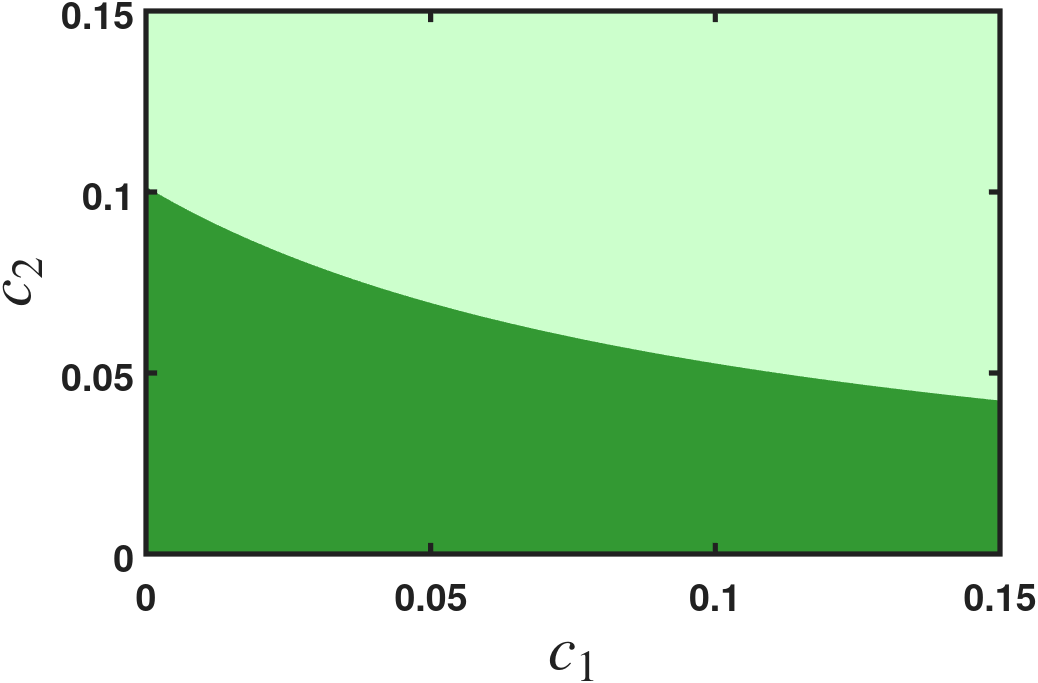}}
}

\caption{(a) $\beta=2 \times 10^{-9}$; (b) $\beta=4 \times 10^{-9}$; (c) $\beta=6 \times 10^{-9}$. The dark green region corresponds to the malware growth region $\mathcal{S}_2$, while the light green region corresponds to the malware extinction region $\mathcal{S}_1$. The rest of the parameters are chosen as mentioned in Table~\ref{tab:parameters}.}
\label{growth_region_plot}
\end{center}
\end{figure}

\section{Hybrid gradient-based global optimization using simulated annealing}\label{sec4}
In this work, we have considered two primary interventions as control strategies: vaccination of susceptible devices ($c_1$) and treatment of infected devices ($c_2$). Each of these measures has an implementation cost. The objective is to determine an optimal combination of these controls that minimizes the total cost, which includes both the infection-related costs and the costs related to control measures. To address this, we formulate the following optimal control problem:

\begin{equation}\label{objective_function}
\min_{0 \leq c_1, c_2 \leq 1} \frac{m_0}{\mathcal{T} \widetilde{N}} \int_0^{\mathcal{T}} I(t; c_1, c_2) dt + k_1 c_1 + k_2 c_2=\min_{0 \leq c_1, c_2 \leq 1} k_0 \int_0^{\mathcal{T}} I(t; c_1, c_2) dt + k_1 c_1 + k_2 c_2,
\end{equation}
where, $k_0=\frac{m_0}{\mathcal{T} \widetilde{N}}>0$, $k_1>0$ and $k_2>0$. $\mathcal{T}>0$ is the maximum time of consideration. Note that the quantity $\frac{1}{\mathcal{T} \widetilde{N}} \int_0^{\mathcal{T}} I(t; c_1, c_2) dt$ signifies the average proportion of the population infected per unit time over the entire duration $[0, \mathcal{T}]$, under the influence of control strategies $c_1$ and $c_2$. $m_0$, $k_1$ and $k_2$ are the associated costs related to the average infection, the vaccination of susceptible and the treatment of infected respectively. Here we derive a gradient-based optimization algorithm to minimize the cost function \cite{anicta2011introduction,barbu2012mathematical}. Let
$$\mathcal{J}(c_1, c_2)=k_0 \int_0^{\mathcal{T}} I(t; c_1, c_2) dt + k_1 c_1 + k_2 c_2.$$
The corresponding adjoint system is given as follows:

\begin{equation}\label{adjoint_eqn_1}
\frac{d \mathcal{H}}{dt}= -A^T \mathcal{H} + (0, 0, 1, 0, 0)^T, \; \mathcal{H}(\mathcal{T})=(0, 0, 0, 0, 0)^T,
\end{equation}
where,
$$A = \left(
\begin{array}{ccccc}
-\beta I -\eta_1 -c_1-\mu & 0 & -\beta S & \sigma_1 & \sigma_2 \\
\beta I  & -\alpha-\eta_2-\mu & \beta S & 0 & 0 \\
0 & \alpha & -c_2 -\mu & 0 & 0 \\
\eta_1 & \eta_2 & c_2 & -\sigma_1-\mu & 0 \\
c_1 & 0 & 0 & 0 & -\sigma_2-\mu
\end{array}
\right),$$
and $S=S(t; c_1, c_2)$, $I=I(t; c_1, c_2)$. Suppose $\mathcal{H}=(\mathcal{H}_1, \mathcal{H}_2, \mathcal{H}_3, \mathcal{H}_4, \mathcal{H}_5)^T$ is a solution of the adjoint system \eqref{adjoint_eqn_1}. 
Denote $C=(c_1, c_2)$. The directional derivative of $\mathcal{J}(C)$ can be defined as
$$d \mathcal{J}(C)(\theta)= \lim_{\epsilon \to 0} \frac{d}{d \epsilon} \mathcal{J}(C+\epsilon \theta),$$
where $\theta=(\theta_1, \theta_2) \in \mathbb{R}^2$, $\epsilon>0$ is an arbitrarily small number such that $C+\epsilon \theta \in [0, 1]^2$.
Now, 
$$ \frac{d}{d \epsilon} \mathcal{J}(C+\epsilon \theta)=k_0 \int_0^{\mathcal{T}} \frac{d I(t, C+ \epsilon \theta)}{d \epsilon} dt + k_1 \theta_1 +k_2 \theta_2.$$

Denote 
\begin{subequations}
    \begin{eqnarray*}
       U(t)&=&(U_1(t), U_2(t), U_3(t), U_4(t), U_5(t))\\
       &=&\lim_{\epsilon \to 0} \frac{d}{d \epsilon}\bigg( S(t, C+ \epsilon \theta), E(t, C+ \epsilon \theta),
       I(t, C+ \epsilon \theta),  R(t, C+ \epsilon \theta), V(t, C+ \epsilon \theta)\bigg). 
    \end{eqnarray*}
\end{subequations}
This implies,
\begin{subequations}\label{cond_J}
    \begin{eqnarray}
      d \mathcal{J}(C)(\theta)&=& k_0 \int_0^{\mathcal{T}} U_3(t) dt + k_1 \theta_1 +k_2 \theta_2.
    \end{eqnarray}
\end{subequations}
From system \eqref{model1} we have

\begin{subequations}\label{eqn_U}
    \begin{eqnarray}
   \frac{d U}{dt} &=& A U +B \theta,
    \end{eqnarray}
\end{subequations}
where, 

$$B = \left(
\begin{array}{cc}
-S & 0  \\
0 & 0  \\
0 & -I  \\
0 & I  \\
S & 0 
\end{array}
\right),\;\;\text{with}\; S=S(t; c_1, c_2) \;\;\text{and}\;\; I=I(t; c_1, c_2).$$
Multiplying both sides of \eqref{eqn_U} by $\mathcal{H}^T$ and integrating with respect to $t$ from $0$ to $\mathcal{T}$ and using the conditions $U(0)=\mathcal{H}(\mathcal{T})=(0, 0, 0, 0, 0)^T$, we obtain

\begin{subequations}\label{cond_1}
    \begin{eqnarray}
   \int_0^{\mathcal{T}} \mathcal{H}^T\frac{d U}{dt} dt &=& \int_0^{\mathcal{T}}\mathcal{H}^T A U dt+  \int_0^{\mathcal{T}}\mathcal{H}^T B \theta dt,\nonumber\\
  \Rightarrow \;\; -\int_0^{\mathcal{T}} \frac{d \mathcal{H}^T}{dt} U dt &=& \int_0^{\mathcal{T}}\mathcal{H}^T A U dt+  \int_0^{\mathcal{T}} \bigg( \theta_1 (\mathcal{H}_5(t; c_1, c_2)-\mathcal{H}_1(t; c_1, c_2)) S(t; c_1, c_2) \nonumber \\
   &&+ \theta_2 (\mathcal{H}_4(t; c_1, c_2)-\mathcal{H}_3(t; c_1, c_2)) I(t; c_1, c_2) \bigg) dt.
    \end{eqnarray}
\end{subequations}

Now taking transpose the adjoint equation \eqref{adjoint_eqn_1}, then multiplyting both sides by $U(t)$ and integrating with respect to $t$ from $0$ to $\mathcal{T}$ we obtain the following relation:

\begin{subequations}\label{adjoint_eqn_2}
    \begin{eqnarray}
   -\int_0^{\mathcal{T}} \frac{d \mathcal{H}^T}{dt} U dt &=& \int_0^{\mathcal{T}} \mathcal{H}^T A U dt -\int_0^{\mathcal{T}} U_3(t) dt.
    \end{eqnarray}
\end{subequations}

Consequently, from \eqref{cond_J}, \eqref{cond_1} and \eqref{adjoint_eqn_2}, the directional derivative of $\mathcal{J}(C)$ is given by:

\begin{subequations}\label{relation_dJ}
    \begin{eqnarray}
   d\mathcal{J}(C)(\theta) &=& \theta_1 \bigg( k_1 - k_0 \int_0^{\mathcal{T}} (Q_5(t; c_1, c_2) -Q_1(t; c_1, c_2) S(t; c_1, c_2) dt\bigg)\nonumber\\
   && + \theta_2 \bigg( k_2 - k_0 \int_0^{\mathcal{T}} (Q_4(t; c_1, c_2) -Q_3(t; c_1, c_2) I(t; c_1, c_2) dt\bigg).
    \end{eqnarray}
\end{subequations}

\begin{center}
\footnotesize
\begin{minipage}{0.95\linewidth}
\begin{algorithm}[H]
    \SetAlgoLined
    \SetKwInOut{Input}{Input}\SetKwInOut{Output}{Output}
    
    \Input{
        Gradient function $\nabla \mathcal{J}(C)$ related to cost, initial value of control $C^{(0)}=(c_1^{(0)}, c_2^{(0)}) \in [0, 1]^2$, \\
        Tolerance levels $\epsilon_k$, $\delta_k$, all the other model parameters in model \eqref{model1}\\
        Parameters related to Simulated Annealing: initial temperature $T$, cooling rate $\lambda$, \\
        number of cooling steps $N_c$, number of perturbations per step $N$
    }

    \Output{Global estimate of optimal control $C^*=(c_1^*, c_2^*)$}
    
    \nl Initialize $C \leftarrow C^{(0)}$, $k \leftarrow 0$

    \nl \Repeat{met the convergence criteria}{

        \nl \textbf{Gradient-based local search}:\\
        \Indp
        \nl Solve the system \eqref{model1} in forward-time to obtain $S(t; C)$, $E(t; C)$, $I(t; C)$, $R(t; C)$ and $V(t; C)$

        \nl Solve the adjoint system \eqref{adjoint_eqn_1} in backward-time to obtain $\mathcal{H}(t; C)$

        \nl Compute the gradient $\nabla \mathcal{J}(C)$ using \eqref{gradient1}

        \nl Perform a local descent step: 
        \[
        \widetilde{C}^{(i)} = C^{(i)} - \eta \cdot \nabla \mathcal{J}(C^{(i)})
        \]

        \nl Project $\widetilde{C}^{(i)}$ on $[0, 1]^2$: 
        \[
        C^{(k)*} = \Pi_{[0, 1]^2}(\widetilde{C}^{(i)})
        \]
        \Indm
        
        \nl \If{$\mathcal{J}(C^{(k)*}) - \mathcal{J}(C^{(k)}) < -\epsilon_k$}{
            $C \leftarrow C^{(k)*}$
        }

        \nl \textbf{Simulated Annealing Phase:} Set temperature $T$

        \nl \For{$j = 1$ \KwTo $N_c$}{

            \nl \For{$i = 1$ \KwTo $N$}{

                \nl Choose $m \in \{1, 2\}$ randomly

                \nl \eIf{$m = 1$}{
                    Re-generate randomly only one entry of $C$ to obtain a new point $\tilde{C}$
                }{
                    Re-generate randomly  both the entries of $C$ to obtain a new point $\tilde{C}$
                }

                \nl Project the randomly generated $\tilde{C}$ on $[0, 1]^2$

                \nl Compute $\Delta \leftarrow \mathcal{J}(\tilde{C}) - \mathcal{J}(C)$

                \nl \eIf{$\Delta < -\delta_k$}{
                    $C \leftarrow \tilde{C}$
                }{
                   Pick a number $r$ between $(0,1)$ randomly from a uniform distribution on $(0, 1)$\\
                    \If{$r < T \exp(-\Delta / T)$}{
                        $C \leftarrow \tilde{C}$
                    }
                }
            }

            \nl Update temperature: $T \leftarrow \lambda T$
        }

        \nl Update $C^{(k+1)} \leftarrow C$, $k \leftarrow k + 1$
    }

    \nl \Return{$C$}
    
\caption{Hybrid gradient-based simulated annealing}
\label{alg:1}
\end{algorithm}
\end{minipage}
\end{center}

Let us define the gradient as follows:

\begin{equation}\label{gradient1}
\nabla \mathcal{J} (C) = \left(
\begin{array}{c}
 k_1 - k_0 \int_0^{\mathcal{T}} (Q_5(t; c_1, c_2) -Q_1(t; c_1, c_2) S(t; c_1, c_2) dt  \\
 k_2 - k_0 \int_0^{\mathcal{T}} (Q_4(t; c_1, c_2) -Q_3(t; c_1, c_2) I(t; c_1, c_2) dt
\end{array}
\right).
\end{equation}

This expression of the directional derivative facilitates the development of a conceptual gradient-based iterative algorithm to determine the optimal control pair \((c_1^*, c_2^*)\). Furthermore, the existence of an optimal control can be established using the approach outlined in~\cite{anicta2011introduction}. 

Now, to identify the global optimum, we employ the hybrid descent algorithm with simulated annealing, a global search algorithm discussed in~\cite{yiu2004hybrid}. The expression derived in \eqref{gradient1} provides the descent direction to reach a local optimum, while the stochastic nature of the simulated annealing approach helps bypass local optima and reach the global optimum. The detailed steps of the algorithm are presented in Algorithm~\ref{alg:1}. The implementation of this algorithm through numerical simulations is presented in the Results section.

\section{Numerical simulation}\label{sec5}
We numerically solve the system \eqref{model1} using the classical fourth-order Runge–Kutta method and investigate various dynamical and control aspects. Unless otherwise specified, all parameter values used in the simulations are taken from Table~\ref{tab:parameters}. The initial values for solving the system are set as: $(S(0), E(0), I(0), R(0), V(0)) = (1 \times 10^9, 0, 1, 0, 0)$. Under this numerical framework, we examine the sensitivity of key parameters, the influence of the transmission rate on epidemic characteristics, the effect of control parameters on infection dynamics, and the evaluation of globally optimal control strategies.

\subsection{Sensitivity analysis}

We calculate the normalized forward sensitivity indices \cite{nadim2021short} w.r.t. to $\mathcal{R}_c$ to have an overview of the most significant parameters that lead to the growing or decaying nature of malware propagation. We consider the parameters: $\beta$, $c_1$ , $c_2$, $\eta_1$, $\eta_2$, $\sigma_1$, $\sigma_2$ and calculate their normalized forward sensitivity indices in $\mathcal{R}_c$, using the values of specific parameters in Table~\ref{tab:parameters}. The forward normalized sensitivity index for a variable $\mathcal{G}$ w.r.t. the parameter $\xi$ is given by:
\begin{equation}\label{sensitivity_formula}
\mathcal{X}_{\mathcal{F}}^{\xi} = \frac{\partial \mathcal{G}}{\partial \xi} \times \frac{\xi}{\mathcal{F}},
\end{equation}
where, $\mathcal{G}$ is explicitly dependent on $\xi$. If $\mathcal{X}_{\mathcal{G}}^{\xi} =k$ then it means that if the value of $\xi$ increases by $1 \%$, then $\mathcal{G}$ will increase by $|k| \%$ (if $k$ is positive) or decrease by $|k| \%$ (if $k$ is negative), while all other parameter values remain fixed. The sensitivity indices are given in Table \ref{sensitivity_table} and in Fig \ref{sensitivity_fig}. The sensitivity indices show that $\sigma_1$, $\sigma_2$, $\beta$ are positively correlated with malware transmission threshold $\mathcal{R}_c$, and $c_1$, $c_2$, $\eta_1$, $\eta_2$ are negatively correlated with $\mathcal{R}_c$. Moreover, $\beta$  (malware transmission rate), $c_1$ (vaccination rate of susceptible devices), $c_2$(treatment to infected devices through reinstallation) and $\eta_1$ (rate of user-initiated device and password resets) appear to be the four most significant parameters.

\begin{figure}[ht!]
\begin{center}
		\includegraphics[scale=0.55]{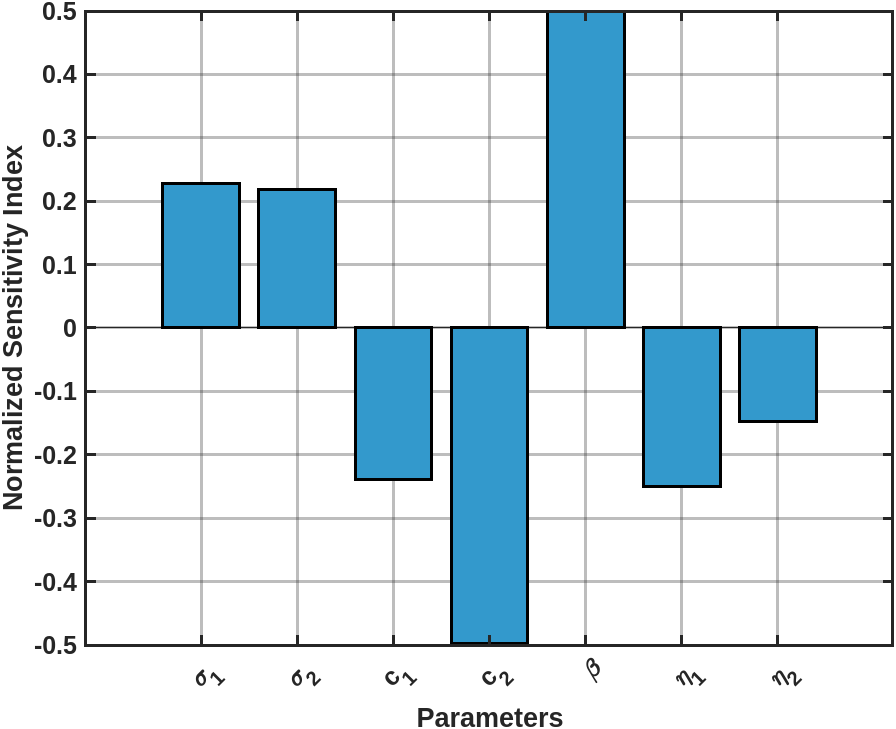}
\caption{Normalized forward sensitivity indices of parameters on malware transmission threshold $\mathcal{R}_c$.}
\label{sensitivity_fig}
\end{center}
\end{figure}

\begin{table}[!ht]
    \centering
    \begin{tabular}{l|l|l|l|l|l|l}
        \hline
        $\mathcal{X}_{\mathcal{R}_c}^{\sigma_1}$ &  $\mathcal{X}_{\mathcal{R}_c}^{\sigma_2}$ & $\mathcal{X}_{\mathcal{R}_c}^{c_1}$&$\mathcal{X}_{\mathcal{R}_c}^{c_2}$ & $\mathcal{X}_{\mathcal{R}_c}^{\beta}$ & $\mathcal{X}_{\mathcal{R}_c}^{\eta_1}$ & $\mathcal{X}_{\mathcal{R}_c}^{\eta_2}$  \\
        \hline
         0.2277 &  0.2186 & -0.2396 & -0.4983 &0.5 & -0.2495 & -0.1469 \\
        \hline
    \end{tabular}
    \caption{Normalized forward sensitivity indices of parameters on $\mathcal{R}_c$.}
    \label{sensitivity_table}
\end{table}

\subsection{Impact of transmission rate}

The transmission rate $\beta$ is one of the most influential parameters, as obtained from the sensitivity analysis. Thus, it is immediate to study the impact of $\beta$ on some key epidemic characteristics. For an infected curve $I(t)$ in a specific time interval $[0, \mathcal{T}]$, the following three important characteristics can quantify the severity of the infection progression:

\begin{figure}[ht!]
\begin{center}
		\mbox{\subfigure[]{\includegraphics[scale=0.6]{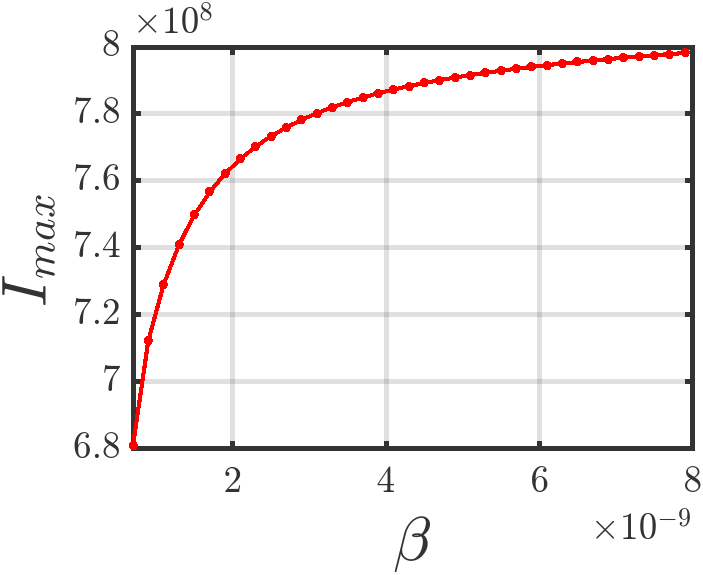}}
		\subfigure[]{\includegraphics[scale=0.6]{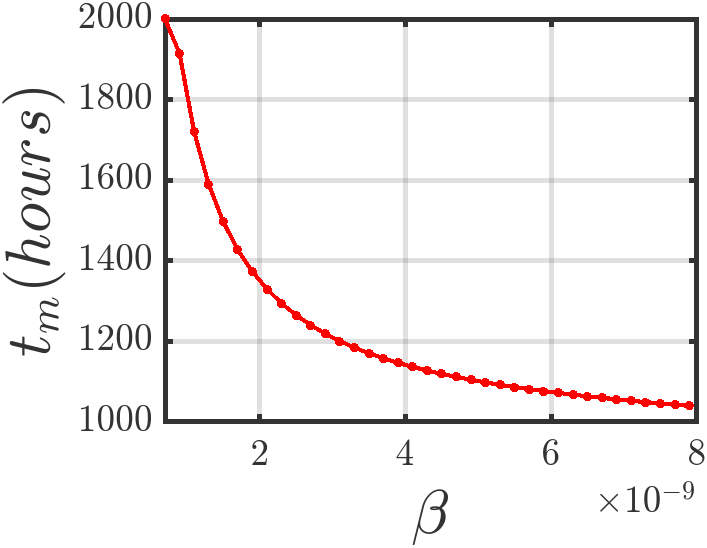}}}
        \mbox{
        \subfigure[]{\includegraphics[scale=0.6]{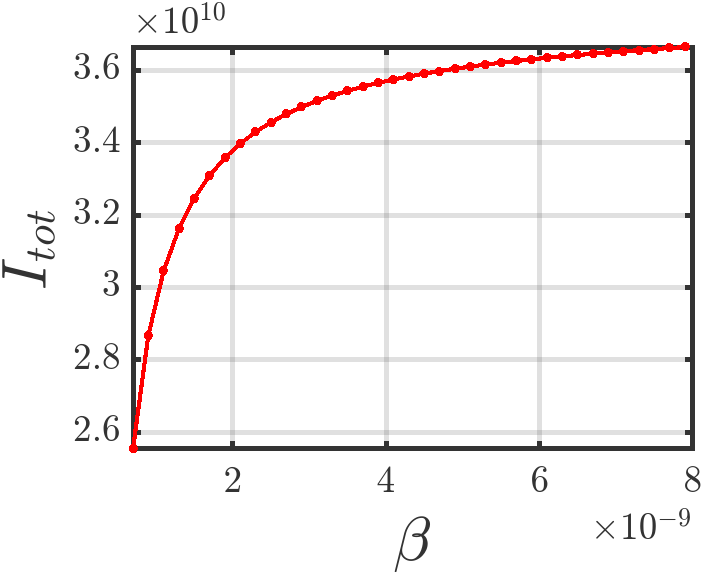}}
		}
\caption{Impact of $\beta$ on (a) maximum number of infected $I_{max}$; (b) time to maximum infected $t_m$; and (c) total number of infected $I_{tot}$. This simulation is carried out with $c_1=c_2=0$ and the rest of the parameters are chosen as mentioned in Table \ref{tab:parameters}.}
\label{beta_vs_I}
\end{center}
\end{figure}

\noindent
\begin{itemize}
    \item {\bf Maximum number of infected.}\;\; Maximum number of infected is denoted by $I_{max}$ and is defined by
$$I_{max}=\max_{t \in [0, \mathcal{T}]} I(t).$$

\item {\bf Time to maximum infected.}\;\; Time to maximum infected is the time when the infected curve $I(t)$ attains its maximum. It is denoted by $t_m$ and defined by
$$t_m = \arg \max_{t \in [0, \mathcal{T}]} I(t).$$

\item {\bf Total number of infected.}\;\; The total number of infected within the interval $[0, \mathcal{T}]$ is denoted by $I_{tot}$ and is defined by
$$I_{tot}= \alpha \int_0^{\mathcal{T}} E(t) dt.$$
Note that the term $\alpha E(t)$ denotes daily number of new devices joining the infected compartment and hence taking integration from $0$ to $\mathcal{T}$ yields $I_{tot}$.
\end{itemize}

Figure~\ref{beta_vs_I} shows the impact of $\beta$ on $I_{\max}$, $t_m$, and $I_{\text{tot}}$. It is observed that $I_{\max}$ increases with an increase in $\beta$; however, the rate of increase in $I_{\max}$ decreases, and $I_{\max}$ reaches saturation for larger values of $\beta$. A similar observation can be drawn for the relationship between $I_{\text{tot}}$ and $\beta$. However, $t_m$ follows an inversely proportional relationship with $\beta$. If the transmission rate increases, then the maximum number of infected individuals will be reached earlier. This observation can indeed help in being better informed about the critical time and staying better prepared accordingly.

\begin{figure}[ht!]
\begin{center}
		\mbox{\subfigure[]{\includegraphics[scale=0.5]{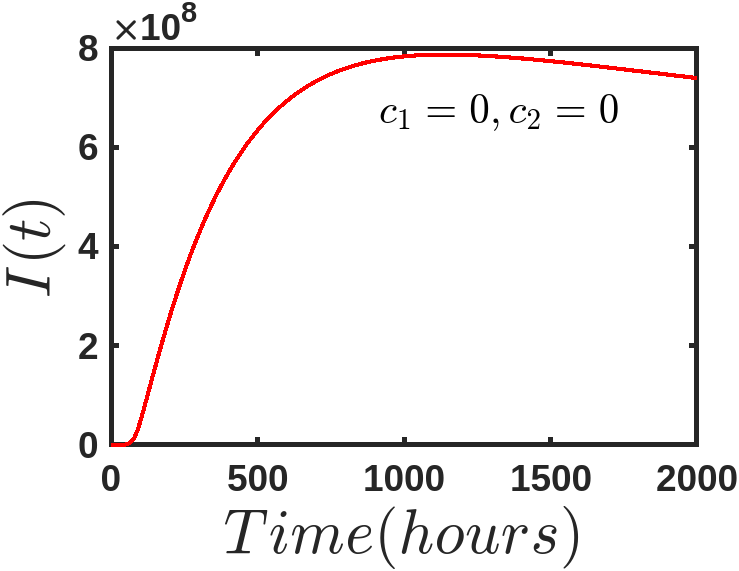}}
		\subfigure[]{\includegraphics[scale=0.5]{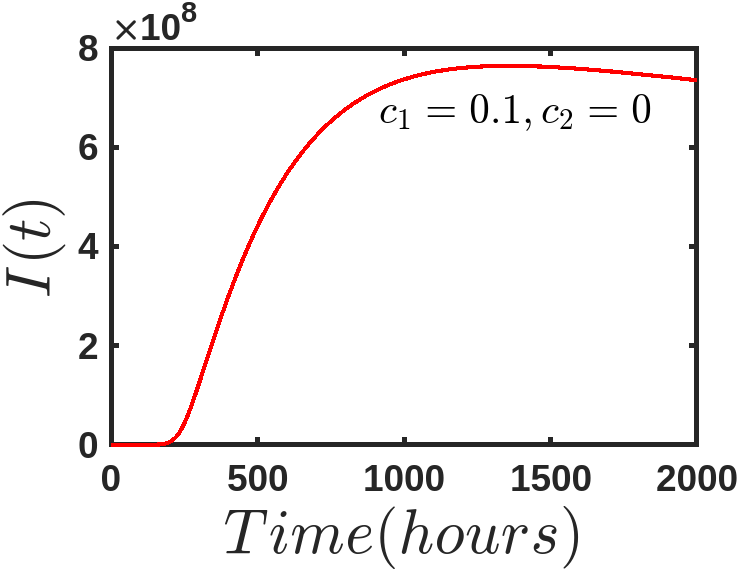}}}
        \mbox{
        \subfigure[]{\includegraphics[scale=0.5]{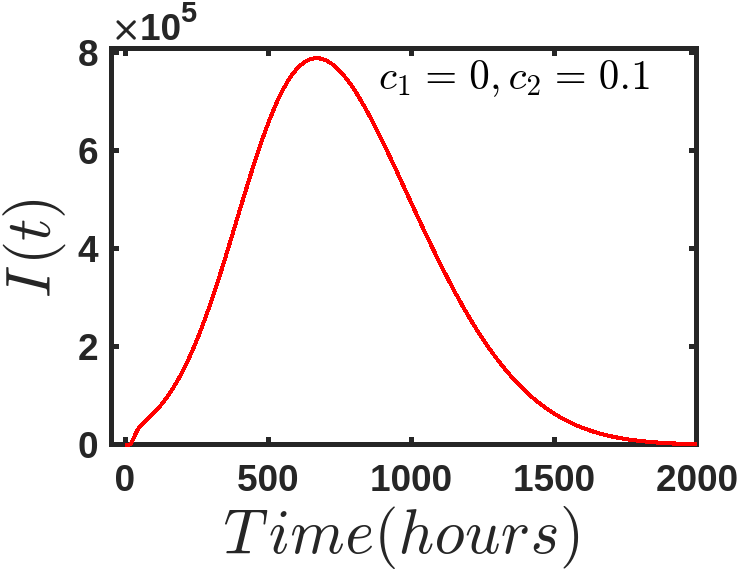}}
		 \subfigure[]{\includegraphics[scale=0.5]{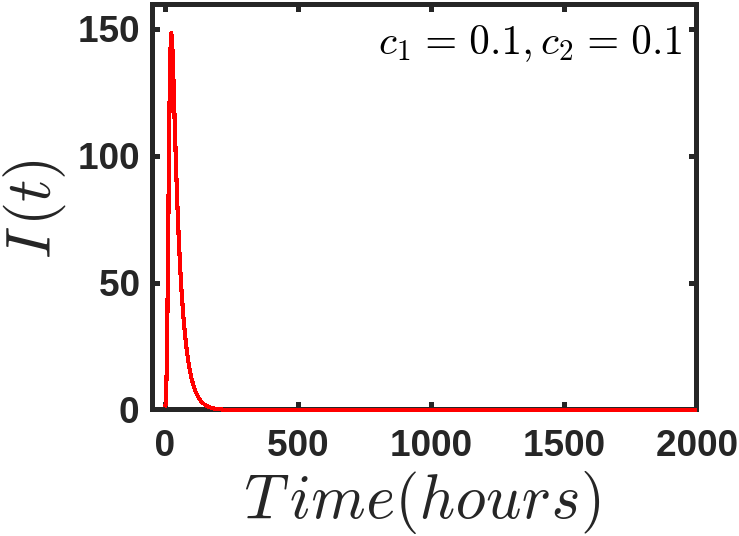}}
         }
\caption{Impact of vaccination ($c_1$) and treatment ($c_2$) on the infected devices counts. The rest of the parameters are chosen as mentioned in Table \ref{tab:parameters}.}
\label{effect_c1_c2}
\end{center}
\end{figure}

\subsection{Combined impact of control measures and transmission rate}

The impact of the strength of different intervention strategies, such as vaccination or treatment, depends strongly on the speed of malware progression. In this subsection, we investigate the combined effects of intervention strengths (e.g., $c_1$, $c_2$) and the transmission rate $\beta$. The results are presented in Figure~\ref{effect_beta_c1_c2}. The upper panel of Figure~\ref{effect_beta_c1_c2} shows the effect of $c_1$ on $I_{max}$, $t_m$, and $I_{tot}$ for different values of $\beta$. For small values of $\beta$, increasing $c_1$ leads to a substantial reduction in the peak infection level $I_{max}$ and the cumulative number of infected $I_{tot}$, while simultaneously delaying the time to peak infection $t_m$. However, as $\beta$ increases, the effectiveness of $c_1$ diminishes, and both $I_{max}$ and $I_{tot}$ remain high even for larger values of $c_1$, indicating that strong transmission can overwhelm this intervention. The lower panel illustrates the impact of $c_2$. Increasing $c_2$ consistently reduces $I_{max}$ across the entire range of $\beta$, with a significant suppression of the infection peak even at higher transmission rates. Moreover, higher values of $c_2$ significantly reduce the total infection burden $I_{tot}$ and shorten the outbreak duration. These results suggest that $c_2$ is relatively more effective in controlling malware spread under high-transmission scenarios, whereas $c_1$ is primarily effective when the transmission rate is relatively low.

\begin{figure}[ht!]
\begin{center}
\mbox{\subfigure[]{\includegraphics[width=0.3\textwidth,height=4.5cm]{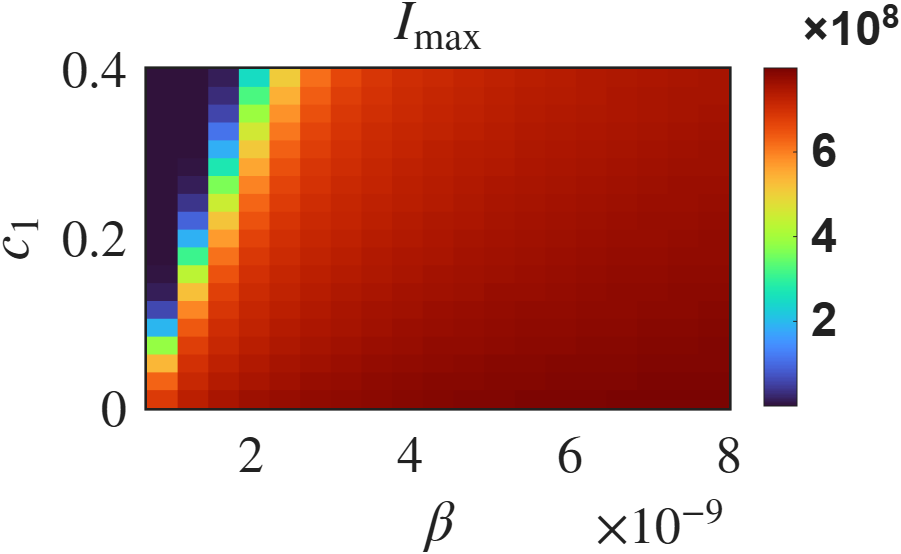}}
		\subfigure[]{\includegraphics[width=0.3\textwidth,height=4.5cm]{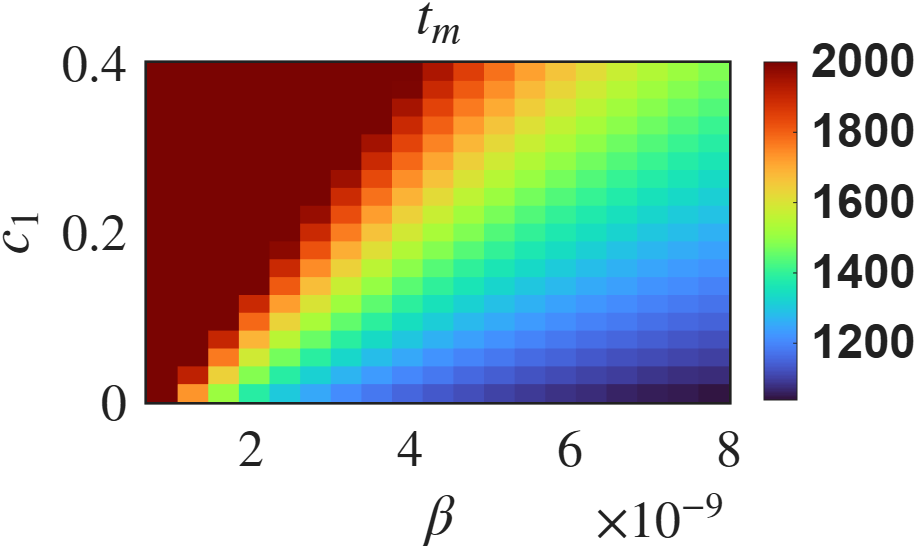}}
        \subfigure[]{\includegraphics[width=0.3\textwidth,height=4.5cm]{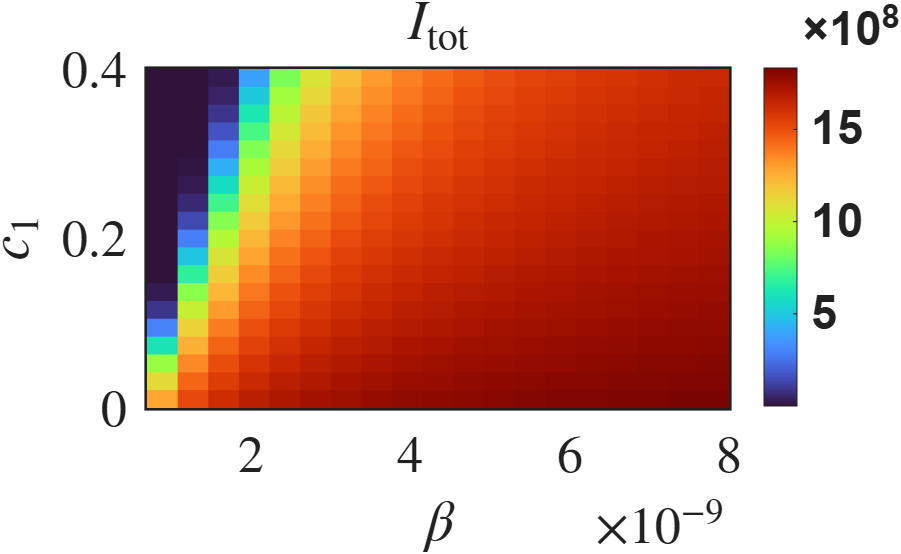}}
        }
		\mbox{\subfigure[]{\includegraphics[width=0.3\textwidth,height=4.5cm]{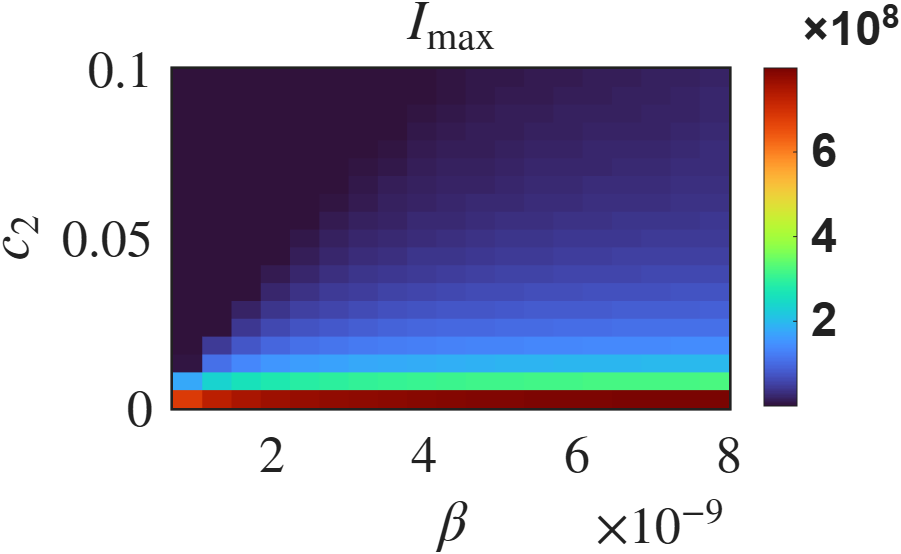}}
		\subfigure[]{\includegraphics[width=0.3\textwidth,height=4.5cm]{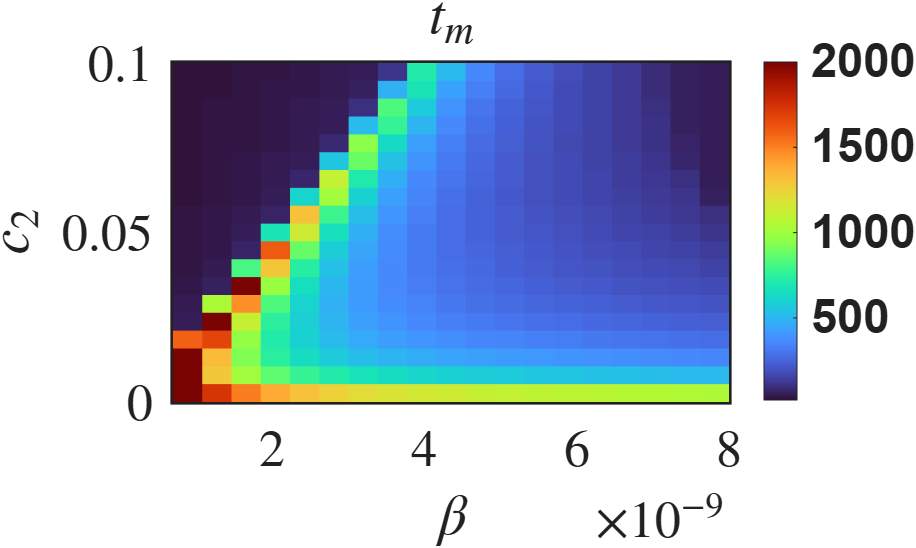}}
        \subfigure[]{\includegraphics[width=0.3\textwidth,height=4.5cm]{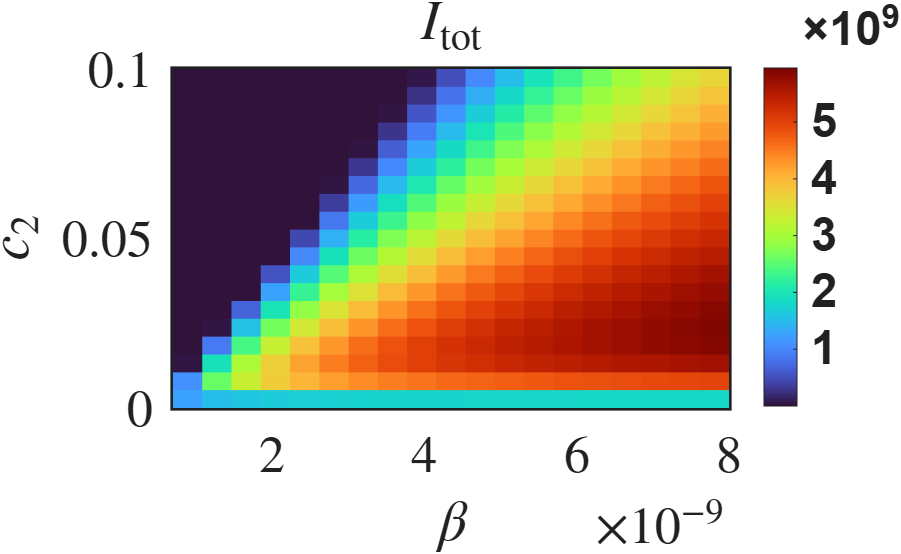}}
        }
       
\caption{Combined impact of the transmission rate $\beta$ and intervention strengths on the malware dynamics. Panels (a)–(c) show the effects of $c_1$ on $I_{\max}$, $t_m$, and $I_{\text{tot}}$, respectively, for varying values of $\beta$. Panels (d)–(f) illustrate the corresponding effects of $c_2$ on $I_{\max}$, $t_m$, and $I_{\text{tot}}$. The rest of the parameters are chosen as mentioned in Table \ref{tab:parameters}. For the upper panel $c_2$ is kept to be $0$ and in the lower panel $c_1$ is kept to be $0$.}
\label{effect_beta_c1_c2}
\end{center}
\end{figure}

\subsection{Impact of vaccination and treatment}

When an epidemic emerges in a network of IoT devices, two immediate possible control measures are the vaccination of susceptible devices and the treatment of infected devices. When it comes to the implementation of both these control measures, a natural question arises: which control measure is most effective, and how does their combination influence the epidemic progression? To answer this, four different scenarios are considered: (i) no vaccination and treatment ($c_1=c_2=0)$; (ii)  only vaccination and no treatment ($c_1=0.1, c_2=0$); (iii) only treatment and no vaccination ($c_1=0, c_2=0.1$); (iv) both vaccination and treatment ($c_1=c_2=0.1$). The outcomes of the infected compartment $I(t)$ under these four scenarios are shown in Figure~\ref{effect_c1_c2}. The results show that without any control measures, the infected count reaches up to $8 \times 10^8$. The administration of only vaccination fails to reduce the infection count significantly, and it remains on the order of $10^8$. In contrast, if only treatment is administered, the infection count effectively reduces to the order of $10^5$. Moreover, if both vaccination and treatment are administered, the epidemic progression is successfully prevented, and the infection count remains below 150. These insights consequently point towards the question of the optimal combination of vaccination and treatment to combat the epidemic.

\subsection{Global optimal allocation of controls}

In this subsection, the performance of Algorithm~\ref{alg:1} is demonstrated through numerical simulation. The original system \eqref{model1} and the adjoint system \eqref{adjoint_eqn_1} are solved using the Runge-Kutta fourth-order method. We assume $\mathcal{T} = 2000$, $m_0 = 1$, $k_1 = 0.2$, and $k_2 = 0.3$. Though this choice of cost parameters is hypothetical and intended for illustration purposes, the algorithm can be implemented similarly with other cost values. However, the values chosen here follow the cost structure $m_0 \geq k_2 \geq k_1$. We normalize the cost of infection to 1, which serves as a baseline reference point for comparing the relative costs of the control parameters. Treating infected devices often involves complex procedures such as malware removal, data recovery, hardware replacement, etc., which may require manual effort and system downtime. In contrast, vaccinating susceptible devices is generally automated and can be administered at scale with minimal cost. Based on this reasoning, we choose the aforementioned cost values. However, one can adjust the cost parameters as per the requirements and reiterate the same algorithm.

\begin{figure}[ht!]
\begin{center}
		\includegraphics[scale=0.28]{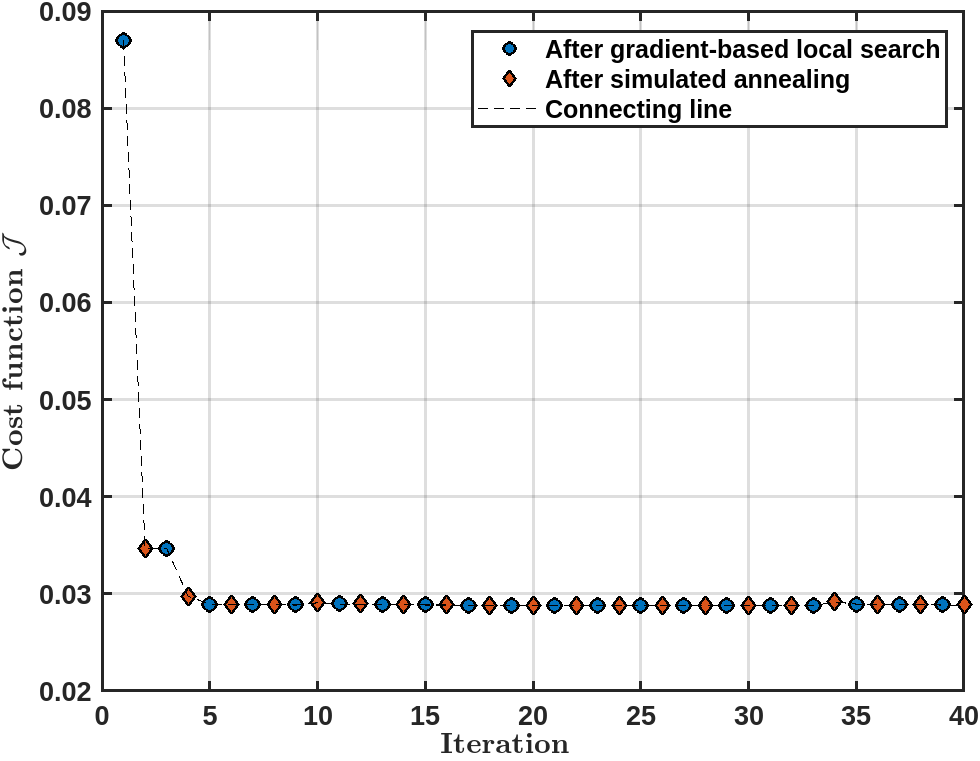}
\caption{A representative convergence trend of the cost function $\mathcal{J}$ over iterations.}
\label{cost_function_fig1}
\end{center}
\end{figure}

We start with an initial choice $(c_1, c_2)=(0.1, 0.1)$ and then find the local minima using the gradient-based search. Then we implement the simulated annealing approach with initial temperature $T=0.02$, cooling rate $\lambda=0.9$. The number of iterations in the simulated annealing part is kept sufficiently large to avoid premature exit from the process.

A representative convergence history of the cost function is shown in Figure~\ref{cost_function_fig1}. The optimal solution appears to be $(c_1^*, c_2^*)=(0.01, 0.08)$ with the minimum value of cost function $\mathcal{J}(c_1^*, c_2^*)=0.028$. Moreover, to ensure the $(0.01, 0.08)$ is the global optimal solution, we started with different points in the $(c_1, c_2)$-plane (e.g., $(0.1, 0.35)$, $(0.25, 0.2)$, $(0.35, 0.1)$, and $(0.1, 0.1)$) and in all the cases the Algorithm \ref{alg:1} produces the same global optimal solution $(0.01, 0.08)$. However, the number of iterative steps to reach the global optimal varies with the starting point. The iterative convergence trajectories for different starting points are shown in Figure~\ref{global_conv_fig2}.

\begin{figure}[ht!]
\begin{center}
		\includegraphics[scale=0.4]{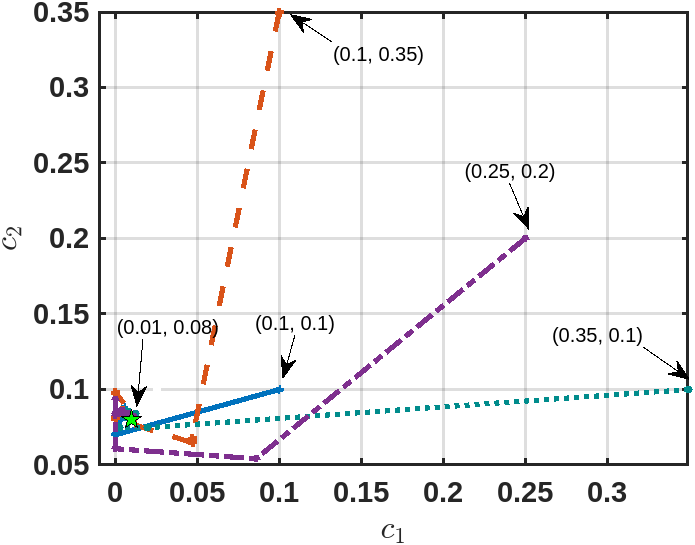}
\caption{Iterative convergence trajectories in the $(c_1,c_2)$-plane for different starting point. The green star  is the global optimal solution $(c_1^*, c_2^*)=(0.01, 0.08)$.}
\label{global_conv_fig2}
\end{center}
\end{figure}

The global optimal solution $(c_1^*, c_2^*) = (0.01, 0.08)$ suggests that vaccinating the susceptible at a rate of $1\%$ and treating the infected at a rate of $8\%$ yields the optimal outcome. To interpret this in terms of total control efforts, let us define the relative control rates as follows:
$$\widehat{c}_1^*=\frac{c_1^*}{c_1^*+c_2^*}\approx 11\%, \;\; \text{and}\;\;\widehat{c}_2^*=\frac{c_2^*}{c_1^*+c_2^*}\approx 89\%.$$
This suggests that $11\%$ of the total control effort should be dedicated to prevention (i.e., vaccination of susceptible devices), and $89\%$ to mitigation (i.e., treatment of infected devices).

\section{Model calibration using Windows malware dataset}\label{section_calibration} 
In this section, we calibrate the proposed model \eqref{model1} by fitting the infection counts to the virus propagation data obtained from the “Windows Malware Dataset with PE API Calls” \cite{catak2021data,kovtun2024entropy}. The transmission rate $\beta$ is estimated using this virus propagation dataset. The model parameters are estimated by minimizing the Sum of Squared Errors (SSE) between the observed epidemic data and the corresponding model outputs, using the SSE objective function
\begin{equation*}
\mathrm{SSE}(\boldsymbol{\theta}) = \sum_{i=1}^{N} \left( y_i - \hat{y}_i(\boldsymbol{\theta}) \right)^2,
\end{equation*}
where $y_i$ denotes the reported number of cases at observation time $i$, $\hat{y}_i(\boldsymbol{\theta})$ is the model prediction depending on the parameter vector $\boldsymbol{\theta}$, and $N$ is the total number of observations. The parameter identification problem is therefore formulated as the unconstrained optimization problem
$\min_{\boldsymbol{\theta}} \; \mathrm{SSE}(\boldsymbol{\theta}).$
\begin{figure}[ht!]
\begin{center}
		\mbox{\subfigure[]{\includegraphics[scale=0.54]{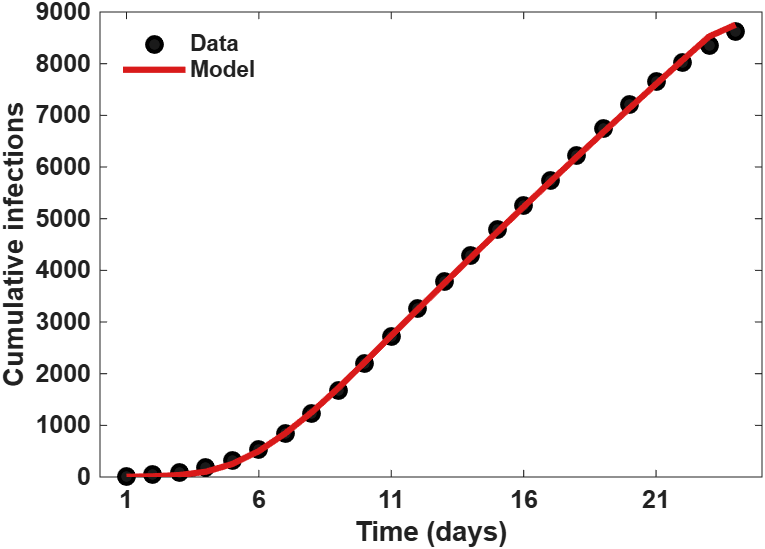}}
		\subfigure[]{\includegraphics[scale=0.54]{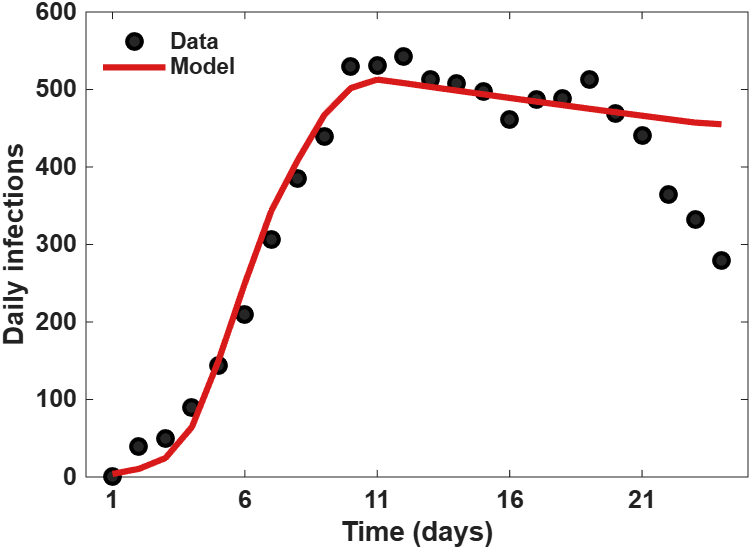}}}
       
\caption{Model fitting with real data of virus propagation. (a) Cumulative data fitting; (b) corresponding daily data fitting.}
\label{fit_1}
\end{center}
\end{figure}

\begin{figure}[ht!]
\begin{center}
		\includegraphics[width=0.8\textwidth,height=6cm]{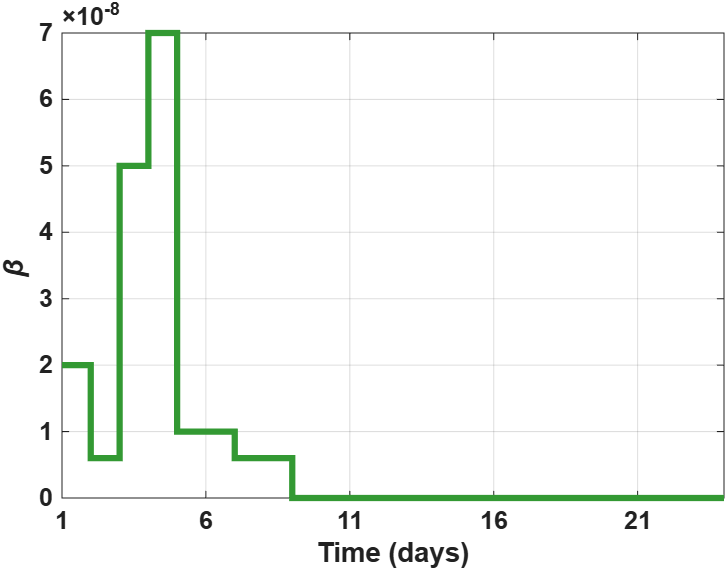}
\caption{Estimated value of $\beta$ from the cumulative data fitting.}
\label{estimated_beta_1}
\end{center}
\end{figure}

\begin{figure}[ht!]
\begin{center}
		\mbox{\subfigure[]{\includegraphics[scale=0.54]{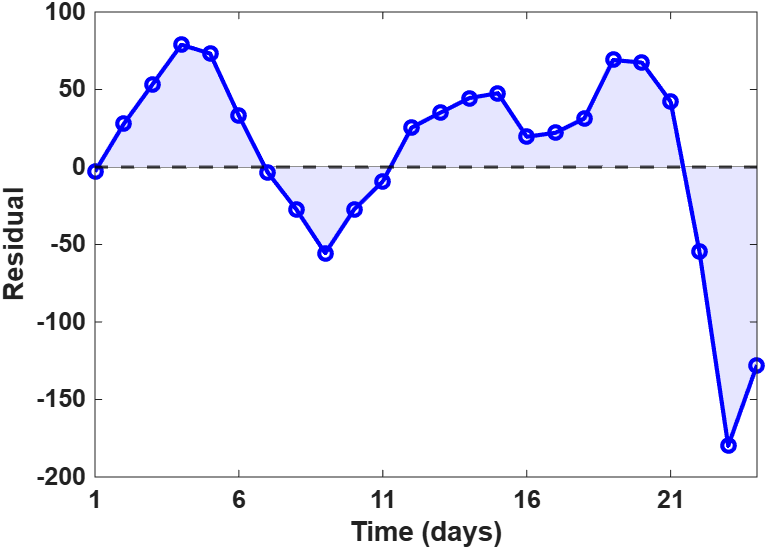}}
		\subfigure[]{\includegraphics[scale=0.54]{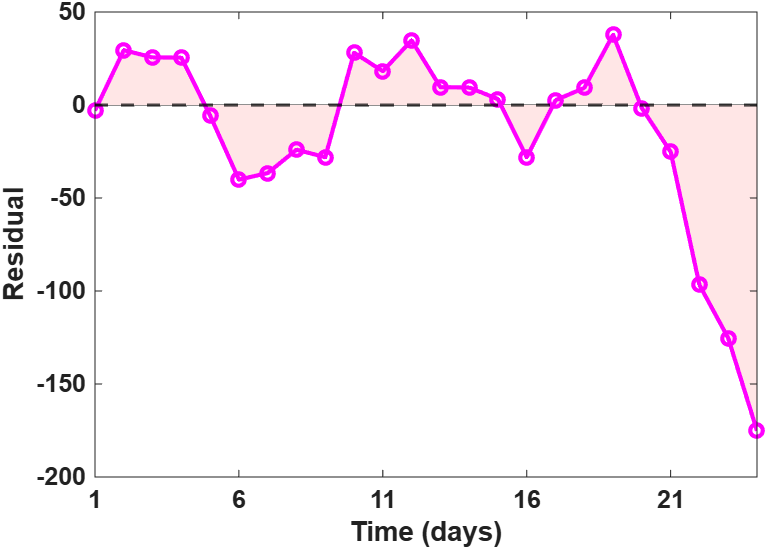}}}
       
\caption{ Residual plots corresponding to the fitting in Figure~\ref{fit_1}: (a) residual for the cumulative data fitting; (b) residual for the daily data fitting.}
\label{residual_1}
\end{center}
\end{figure}

\begin{figure}[ht!]
\begin{center}
		\mbox{\subfigure[]{\includegraphics[scale=0.54]{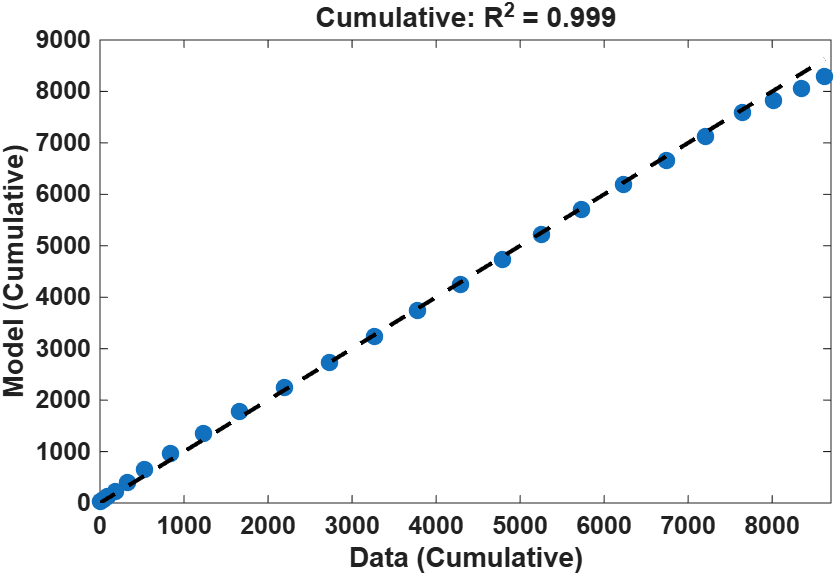}}
		\subfigure[]{\includegraphics[scale=0.54]{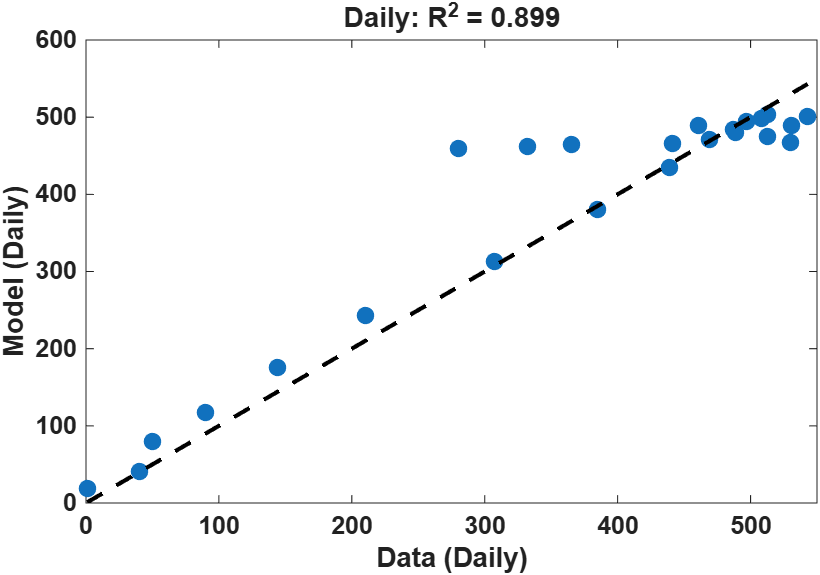}}}
       
\caption{ Plots of the $R^2$ values corresponding to the fitting in Figure~\ref{fit_1}: (a) the $R^2$ value for the cumulative data fitting; (b) the $R^2$ value for the daily data fitting.
}
\label{R2_1}
\end{center}
\end{figure}

\begin{figure}[ht!]
\begin{center}
		\includegraphics[width=0.8\textwidth,height=6cm]{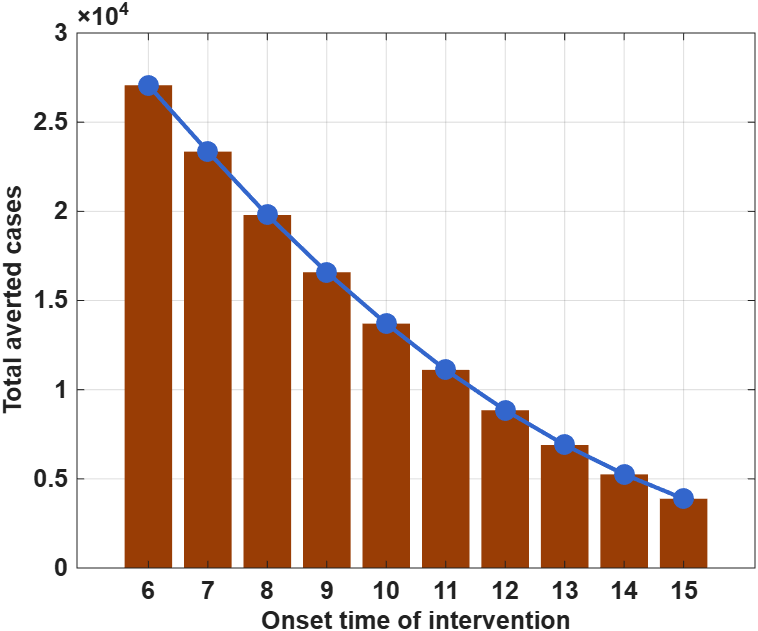}
\caption{Plot of the total number of averted cases as a function of the intervention onset time.}
\label{averted_1}
\end{center}
\end{figure}

To solve this problem, a hybrid optimization strategy is adopted in which an initial gradient-based search is followed by a derivative-free refinement using the MATLAB routine \texttt{fminsearch} \cite{kelley1999iterative}. The solver \texttt{fminsearch} is based on the Nelder--Mead simplex method and seeks a local minimizer of a scalar objective function
$
\min_{\boldsymbol{\theta} \in \mathbb{R}^p} \; f(\boldsymbol{\theta}),
$
without requiring derivative information. At each iteration, a simplex composed of $p+1$ candidate points is updated by geometric transformations such as reflection, expansion, contraction, and shrinkage. For instance, the reflected point is computed as
$$
\boldsymbol{\theta}_r = \boldsymbol{\theta}_c + \alpha \left( \boldsymbol{\theta}_c - \boldsymbol{\theta}_h \right),
$$
where $\boldsymbol{\theta}_c$ denotes the centroid of the simplex excluding the worst vertex $\boldsymbol{\theta}_h$, and $\alpha>0$ is the reflection coefficient. The iteration continues until prescribed convergence criteria on the objective value or parameter variation are satisfied. The results of the model fitting are presented in Figure~\ref{fit_1}. The black dots represent the observed data, while the red curves denote the model-predicted trajectories. The model was fitted to the cumulative infection data in order to estimate the transmission rate $\beta$, and the corresponding fit is shown in Figure~\ref{fit_1}(a). The associated fit to the daily infection data is displayed in Figure~\ref{fit_1}(b). The cumulative infection curve is reproduced reasonably well by the model; however, some discrepancies are observed in the daily infection curve toward the end of the observation period. The estimated time-varying transmission rate $\beta(t)$ is illustrated in Figure~\ref{estimated_beta_1}. It can be seen that $\beta$ initially increases, reaches a peak around $t=5$, and then gradually decreases thereafter. The goodness-of-fit analysis, including the residual plots and the coefficient of determination ($R^2$), is shown in Figure~\ref{residual_1} and Figure~\ref{R2_1}, respectively. Due to the lack of intervention-related data in this fitting, we did not consider interventions explicitly, and their average impact was assumed to be captured by the weekly varying transmission rate. However, we use the calibrated model to assess the qualitative impact of the onset time of intervention strategies. Figure~\ref{averted_1} shows the total averted cases as a function of the time of onset of intervention. It shows an exponentially decaying relationship between the averted cases and the time of introduction of the interventions.


\section{Conclusion}\label{sec6}
Mathematical models are a valuable way to study how malware spreads and to evaluate possible measures for controlling it. In this work, we propose a generic SEIR-type epidemic model using a system of ODEs to describe the characteristics of malware propagation in an IoT network. We proved the positivity and boundedness of the system. Furthermore, we analytically derived expressions for the malware propagation threshold and the separatrix of epidemic regions in the control space. We also studied the local asymptotical stability and global stability of the malware-free equilibrium point. Furthermore, we proved the existence of endemic equilibrium when $\mathcal{R}_0>1$ and we observed the existence of forward bifurcation in the system \eqref{model1}. Using normalized forward sensitivity indices, we identified that the transmission rate, the two control parameters, and the rates of user-initiated device reset and credential reset are the four most sensitive parameters influencing the malware propagation threshold. 

We also studied the explicit dependency of the transmission rate ($\beta$) on the maximum number of infected ($I_{\text{max}}$), the time to maximum infected ($t_m$), and the total number of infected ($I_{\text{tot}}$). It is observed that all three quantities follow a nonlinear relationship with $\beta$. Initially, $I_{\text{max}}$ and $I_{\text{tot}}$ increase with $\beta$, then reach a plateau beyond certain values of $\beta$, whereas $t_m$ exhibits an inverse nonlinear relationship with $\beta$. These trends can be explained by the fact that, at lower values of $\beta$, an increase in the transmission rate accelerates the spread, leading to higher peaks ($I_{\text{max}}$) and larger epidemic sizes ($I_{\text{tot}}$). However, beyond a threshold, the susceptible population becomes depleted more rapidly, limiting further growth and causing the curves to saturate. Conversely, as $\beta$ increases, the infection spreads faster, which shortens the time required to reach the peak ($t_m$), thereby explaining the inverse nonlinear relationship. These explicit relationships shed light on the infection curve and help in understanding its dynamics more clearly.

We proposed a hybrid gradient-based global optimization control algorithm for the SEIRV-type model. We employed the concept of simulated annealing, which, to the best of our knowledge, has not previously been applied to the control of SEIRV-type multi-compartmental epidemic models. We explicitly derived the gradient-descent direction for the cost function related to the SEIRV model for local search of optimum and then used the stochastic search through simulated annealing to reach the global optimum. It is noteworthy to mention that the derivation of our algorithm can be easily implemented on any multi-compartmental epidemic model described by ODEs, however, the derivation of gradients might differ across models. This global search can somehow help in obtaining better control strategies.

The implementation of the optimization algorithm is subject to specific prerequisites, such as the cost coefficients of infection and controls, which may be fixed based on the scenario. However, for illustration purposes, we showed that if $m_0 = 1$, $k_1 = 0.2$, and $k_2 = 0.3$ (following the cost structure $m_0 \geq k_2 \geq k_1$), then ($c_1^*, c_2^*$) = $(0.01, 0.08)$, i.e., a combination of both treatment and vaccination appears to be the global optimal control. In terms of total control efforts, approximately $11\%$ should be dedicated to vaccination and approximately $89\%$ to treatment.

The present study offers a theoretical foundation for understanding IoT malware dynamics through an SEIRV-based formulation and demonstrates how classical epidemic ideas can be adapted to cybersecurity settings. As a natural direction for future work, the model can be enriched by incorporating additional features of modern IoT ecosystems, such as device heterogeneity, multi-layered network structures, and varying patch or update cycles. Real IoT malware often exhibits adaptive and device-specific behaviour; therefore, extending the model to include time-dependent or state-dependent parameters would make the framework more flexible. Another promising avenue is the integration of empirical information—such as network logs, threat-intelligence data, or expert-informed cost coefficients—to calibrate model parameters and further enhance interpretability. The cost coefficients ($m_0$, $k_1$, and $k_2$) can be estimated more realistically by incorporating economic information, system-level implementation assessments, and expert input to better reflect the practical cost-effectiveness of controlling malware spread. These developments would help bridge theoretical insights with real-world IoT security practices, strengthening the practical relevance of the modelling framework.

\section*{Data availability} The sources of all data used/mentioned during this study are cited in the article. 

\section*{Acknowledgements}

This research is partly funded by Ministry of Electronics and Information Technology (MeitY), Government of India, under the Cybersecurity Research \& Development Scheme.

\section*{Author Contributions Statement} SG and VAK jointly formulated the problem. SG performed the analytical analysis and numerical simulations and wrote the main manuscript. VAK contributed to the interpretation of the results and supervised and reviewed the work.

\section*{Conflict of interest}
The authors do not have any conflicts of interest.



 \bibliographystyle{plain} 
 \bibliography{reference}

\end{document}